\documentclass{article} 
\usepackage{iclr2027_conference,times}

\usepackage{signalpreamble}
\usepackage{tikz}
\usepackage{booktabs}
\usepackage{bm}

\usetikzlibrary{arrows.meta,calc,fit,positioning,shapes.geometric}

\providecommand{\St}{\operatorname{St}}
\providecommand{\diag}{\operatorname{diag}}
\providecommand{\wrap}{\operatorname{wrap}}
\providecommand{\Exp}{\operatorname{Exp}}
\providecommand{\Log}{\operatorname{Log}}

\usepackage[most]{tcolorbox}
\newlength{\boxunit}
\newtcolorbox{lightgrey}{
  colback=gray!10, 
  colframe=gray!60,
  boxrule=0.2pt,      
  auto outer arc,     
  width=\linewidth,
  left=1mm,
  right=1mm,
  top=0.3\boxunit,
  bottom=0.3\boxunit,
  before skip=5\boxunit,
  after skip=5\boxunit,
  fontupper=\small,
  before upper=\setlength{\baselineskip}{0.5\baselineskip}
}

\newtcolorbox{greybox}{
  boxsep=0mm, 
  top=0.4\boxunit,
  bottom=0.4\boxunit,
  before skip=0.4\boxunit,
  after skip=0.4\boxunit,
  width=\linewidth,
  left=2mm,
  right=2mm,
  fontupper=\small
}

\title{PolarFlow: Spectral Geometric Flow Matching for Directed Graph Generation via Polar Decomposition}

\author{%
\hspace*{-\tabcolsep}%
\makebox[\textwidth][c]{%
\begin{tabular}{@{}c@{}}
\textbf{Yanan Zhao$^{1}$, Feng Ji$^{1}$,
Alba Carballo-Castro$^{2}$, Jingyang Dai$^{1}$,} \\
\textbf{Antonio Ortega$^{3}$, Pascal Frossard$^{2}$
and Wee Peng Tay$^{1}$}
\\[0.5em]
{\normalfont\small $^{1}$Nanyang Technological University, Singapore} \\
{\normalfont\small $^{2}$LTS4, EPFL, Lausanne, Switzerland} \\
{\normalfont\small $^{3}$University of Southern California, USA}
\end{tabular}%
}%
}
\newcommand{\bestcell}[1]{%
  \cellcolor{gray!25}\(\bm{#1}\)%
}

\iclrfinalcopy 
\begin{document}

\maketitle

\begin{abstract}
Directed graphs model asymmetric, ordered relationships in biology, transportation, social networks, and electronic design automation, making their generation important for analysis, simulation and design. Recent discrete diffusion and flow-matching methods generate directed graphs effectively, but their edge-wise categorical formulations model every ordered node pair, requiring dense predictions during training and repeated categorical updates during generation. This motivates a fundamental question: can comparable generative quality be achieved through a more computationally efficient formulation that better exploits the structure of directed adjacency matrices? Moreover, existing directed-graph benchmarks cover a limited range of network structures, hindering a comprehensive evaluation of graph generation.
We introduce \emph{PolarFlow}, a continuous geometric flow-matching framework based on the polar decomposition of the adjacency matrix and truncated spectral coordinates. Its representation consists of four coupled coordinate components: a Stiefel frame and rotation angles that encode the orthogonal factor, together with a second Stiefel frame and nonnegative eigenvalues that characterize the positive-semidefinite factor. PolarFlow learns a joint velocity field along geometry-aware conditional paths on this product space, reconstructs a continuous adjacency-score matrix, and discretizes it into a binary directed graph. By replacing discrete edge-wise generation with joint continuous transport on this product space, PolarFlow offers a computationally efficient alternative while retaining the global structure of the adjacency matrix. To broaden evaluation, we introduce complementary Circuit and Twitter benchmarks spanning engineered functional and naturally evolved social networks. Across synthetic and real-world datasets, PolarFlow achieves generative quality comparable to competitive discrete methods while substantially reducing training and sampling time.
\end{abstract}

\section{Introduction}
\label{sec:introduction}

Directed graphs naturally represent systems in which the direction of information, influence, or material flow is semantically consequential, including applications in biology
\citep{wei2024inference,takane2025directed,li2006directed}, electronic circuits \citep{liu2025syncircuit}, transportation \citep{concas2022chained}, and social networks \citep{schweimer2022generating}. 
Because edge orientation affects reachability, propagation, and feedback, generating realistic directed graphs supports simulation and data augmentation when observations are scarce or costly, as well as design-space exploration under domain-specific constraints, such as neural architecture search and circuit synthesis \citep{zhang2019dvae,liu2025syncircuit,carballocastro2026directo,li2025layerdag}.

Existing approaches span three broad strategies. First, \emph{latent autoencoding with autoregressive decoding} encodes a directed acyclic graph (DAG) via asynchronous message passing and decodes it node-by-node \citep{zhang2019dvae}. Second, \emph{layerwise autoregressive diffusion}, as in LayerDAG \citep{li2025layerdag}, exploits DAG topology to generate each layer and its incoming edges. Third, \emph{one-shot generation over categorical edge variables}, as in Directo \citep{carballocastro2026directo}, extends discrete flow matching to directed graphs using direction-aware attention and asymmetric positional encodings. A related continuous latent alternative instead applies flow matching to node-level representations learned by a separately trained graph autoencoder \citep{siraudin2026principledlatentdiffusiongraphs}. Further discussion and comparison of related methods are provided in \cref{app.related_works}.

Each paradigm entails distinct computational or representational trade-offs. Autoregressive and layerwise methods generate nodes or topological layers sequentially, limiting parallelism; their reliance on topological structure also restricts them to DAGs.
Discrete flow-matching methods avoid this sequential dependence by predicting all edge variables in parallel at each step, but require dense outputs over all $N(N-1)$ ordered node pairs. With $\calC_E$ denoting the number of edge categories, including the absent-edge category, each network evaluation produces $\mathcal{O}(N^2\calC_E)$ edge logits, accumulating to $\mathcal{O}(TN^2\calC_E)$ over $T$ sampling steps, excluding the costs of the denoising backbone and structural encodings. Continuous latent methods reduce repeated adjacency-space prediction but require a separately trained encoder--decoder. These trade-offs motivate a central question: can a continuous representation derived directly from the adjacency matrix capture orientation-sensitive structure and achieve comparable generative quality more \emph{efficiently}?

\begin{figure}[t]
    \centering
    \includegraphics[width=\textwidth]
    {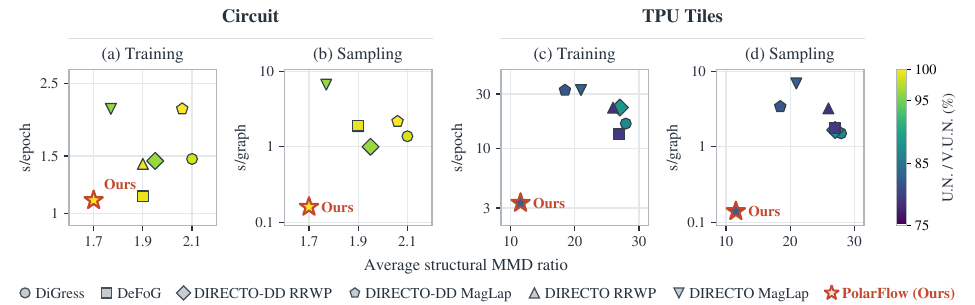}
    \vspace{-6mm}
    \caption{\small Quality-efficiency comparisons on Circuit and TPU Tiles.
    Lower-left is better.
    }
    \label{fig:quality_efficiency}
    \vspace{-5mm}
\end{figure}

To answer this question, we seek a representation of directed graphs that enables continuous flow matching while retaining the global structure of the adjacency matrix. We achieve this through polar decomposition, which separates a directed adjacency matrix into orientation and magnitude components, represented in a compact set of spectral and geometric coordinates. Building on this representation, we introduce \emph{PolarFlow}, which learns a continuous velocity field specifying how the orientation and magnitude coordinates evolve together along paths from source to data representations; the adjacency matrix is reconstructed only at the end of the generation process. This design replaces edge-wise categorical prediction and iterative discrete updates with continuous transport in a lower-dimensional structured space, directly targeting the computational bottleneck of existing discrete approaches. Importantly, our goal is not to trade generative quality for efficiency: as shown in \cref{sec:experiments} (previewed in \cref{fig:quality_efficiency}), PolarFlow achieves quality comparable to that of strong discrete flow-matching methods while substantially reducing both training and sampling times.

 
Our contributions are summarized as follows.
\begin{itemize}
    \item[\textbf{(i)}]  We introduce a compact polar--spectral representation for general directed graphs that separates orientation-sensitive alignment from nonnegative spectral magnitude. Its truncated coordinates comprise two Stiefel frames, rotation angles, and nonnegative eigenvalues.
    \item[\textbf{(ii)}]  We develop geometry-aware conditional flow matching that jointly evolves these four channels, together with a reconstruction and discretization procedure for generating binary directed structures, and a separate process for categorical node-label generation.
    \item[\textbf{(iii)}]  We establish the local perturbation stability of PolarFlow under spectral and geometric regularity, proving that small perturbations of a regularized adjacency endpoint induce proportionally small changes in the conditional flow and learned velocity predictions.
    \item[\textbf{(iv)}] We introduce new \emph{Circuit and Twitter datasets}, representing engineered functional networks and naturally formed social networks, respectively. We demonstrate that PolarFlow achieves competitive generative quality while delivering substantial empirical efficiency gains across both synthetic and real-world directed graphs.
\end{itemize}

\section{Problem Formulation and Preliminaries}
\label{sec:preliminaries}
\vspace{-2mm}
\paragraph{Directed graphs and conditional generation}

Let $G=(\calV,\bA,\by)$ be a directed graph with vertex set $\calV=\{1,\ldots,N\}$, binary adjacency matrix $\bA\in\{0,1\}^{N\times N}$, and optional categorical node labels $\by\in\calC^{N}$, where $\calC$ is finite. We use $\bA_{ij}=1$ to indicate the directed edge $i\to j$ and assume $\bA_{ii}=0$ unless stated otherwise. Given training graphs, our goal is to learn the adjacency distribution $p_\phi(\bA\mid N)$ conditioned on graph size and, when node labels are available, the label distribution 
$p_\psi(\by \mid \bA)$ conditioned on the adjacency matrix, where $\phi$ and $\psi$ denote the parameters of the structural and label models, respectively.
At generation time, $N$ is sampled from the empirical distribution of graph sizes in the training set and held fixed throughout structural generation; for labeled datasets, node labels are then sampled conditional on the generated adjacency matrix.
\vspace{-1mm}
\paragraph{Polar decomposition and spectral representations}

Every real square matrix $\bA\in \Real^{N \times N}$ admits a right polar decomposition (see \citet{kwak2024frequency,ji2026digraph} for the graph signal processing perspective)
\begin{equation}
\label{eq.polar_decomposition}
    \bA=\bU\bP, 
\end{equation}
where $\bU \in \Real^{N \times N} $ is orthogonal and $\bP = (\bA^*\bA)^{1/2} \succeq 0$\footnote{In this paper, we use $^*$ for the Hermitian transpose, which is the usual matrix transpose for real matrices.} is symmetric positive semi-definite. The symmetric factor admits the eigendecomposition
\begin{equation}
\bP=\bQ_p\bLambda_p\bQ_p^*, \qquad \bLambda_p=\diag(\lambda_{p,1},\dots, \lambda_{p,N}),
\end{equation}
where $\bQ_{p}$ is orthogonal and $\lambda_{p,1}\geq \lambda_{p,2} \geq \cdots \geq \lambda_{p,N} \geq 0$. The orthogonal factor $\bU$ is normal and also admits an unitary eigendecomposition, whose eigenvalues lie on the complex unit circle. Specifically, we have
\begin{equation}
\bU= \bQ_u \bLambda_u\bQ_u^*,
\label{eq:orthogonal_spectral}   
\end{equation}
where the diagonal matrix $\bLambda_u$ consists of eigenvalues of the form $e^{\mathrm i\theta_i}, \theta_i\in [0,2\pi)$ for $1\leq i\leq N$, which is either $\pm 1$ or form complex conjugate pairs. 
Thus, the polar decomposition represents $\bA$ through orthonormal/unitary frames, circular rotation angles, and nonnegative spectral values.
\vspace{-1mm}
\paragraph{Stiefel manifolds and Riemannian flow matching} 
The spectral frames $\bQ_{p}$ and $\bQ_{u}$ consist of orthonormal columns. An orthonormal or unitary $k$-frame belongs to the Stiefel manifold \citep{Hat02,Abs08}
\begin{align*}
    \St(N,k)=\set*{\bS\in\mathbb M^{N\times k} \given \bS^* \bS=\bI_k},
\end{align*}
with $\mathbb M$ either $\Real$ (for orthogonal matrices) or $\mathbb C$ (for unitary matrices).
The full orthonormal basis corresponds to the special case $k=N$. The tangent space at a point $\bS \in \St(N,k)$ is
\begin{align*}
\calT_{\bS}(\St(N,k))=\set*{\bV\in\mathbb M^{N\times k} \given \bS^* \bV+\bV^* \bS=0}.  
\end{align*}
The Stiefel manifold is equipped with the Riemannian metric $g_{\bS}(\bV_{1},\bV_{2})=\text{tr}(\bV_{1}^*\bV_{2})$
for $\bV_1,\bV_2 \in \calT_{\bS}(\St(N,k))$. \emph{Note}: for convenience, we do not distinguish real and complex Stiefel manifolds in notation if no confusion arises.

Riemannian flow matching generalizes flow matching to a Riemannian manifold $(\calM, g)$, where $\calM$ denotes the state space and $g_{Z}$ defines an inner product on each tangent space $\calT_{Z}\calM$. We adopt its conditional formulation, which constructs a path for each source–data endpoint pair. Given a source endpoint $Z_0$ and a data endpoint $Z_1$, let $Z_{t}\in \calM$ be a differentiable conditional path satisfying $Z_{t=0}=Z_0$ and $Z_{t=1}=Z_1$ for $t\in[0,1]$.
Its velocity $\dot Z_t=\mathrm dZ_t/\mathrm dt\in \calT_{Z_t}\mathcal M$ provides a tractable regression target for a time-dependent vector field $v_\phi(Z,t)\in \calT_Z\calM$, with the prediction error measured using $g_{Z_t}$.
Because $Z_t$ and $\ud Z_t/\ud t$ are analytically available during training, the learned ODE need not be integrated and is solved only during generation.

\section{Methodology: PolarFlow}
\label{sec:method}
In this section, we introduce PolarFlow, a geometry-aware flow-matching framework that maps directed adjacency matrices into four coupled polar–spectral channels: two Stiefel frames, rotation angles, and nonnegative eigenvalues. A shared time-dependent vector field is trained jointly along channel-specific conditional paths in the resulting product space. We establish local perturbation stability, describe reconstruction and discretization to obtain a binary directed adjacency matrix, and use a separate structure-conditioned model to generate node labels for labeled datasets. The complete PolarFlow pipeline is illustrated in \cref{fig:polarflow_overview}.

\subsection{Four-channel polar--spectral representation}
\label{sec:representation}

Rather than modeling the $N^{2}$ entries of a directed adjacency matrix directly, PolarFlow constructs a compact geometric representation with four polar--spectral channels. It obtains these channels by applying the polar decomposition in \cref{eq.polar_decomposition} to separate the orientation and scaling information in the input matrix. For a data adjacency matrix $\bA$, the representation is computed from the perturbed matrix $\bA_\varepsilon=\bA+\bE_\varepsilon$, where the perturbation satisfies
$\lVert \bE_\varepsilon\rVert_2\leq\varepsilon$. This perturbation is introduced to mitigate potential rank and spectral degeneracies; setting $\varepsilon=0$ recovers the unperturbed endpoint.

For the scaling channel, let $\lambda_{p,1}\geq \lambda_{p,2} \geq \cdots \geq \lambda_{p,N} \geq 0$ denote the eigenvalues of $\bP$, with corresponding orthonormal eigenvectors $\bq_{p,1},\ldots, \bq_{p,N}$. Retaining the $k$ leading eigenpairs gives the rank-$k$ approximation 
\begin{align}
    \bP_{k} = \bQ_{p}\bLambda_{p}\bQ_{p}^{*} \qquad \bQ_{p} = [\bq_{p,1},\ldots, \bq_{p,k}] \in \St(N,k) 
\end{align}
where $\bLambda_{p}=\diag(\blambda_p^{(k)})$ with $\blambda_p^{(k)}=(\lambda_{p,1},\cdots, \lambda_{p,k}) \in \Real^k_{+}$. Thus, $\bQ_{p}$ and $\blambda_p^{(k)}$ constitute the two scaling channels. To construct the two orientation channels $\bQ_{u}$ and $\btheta^{(k)}$, we use the spectral representation of the orthogonal factor $\bU$ in \cref{eq:orthogonal_spectral}. Let $(e^{\mathrm i\widetilde{\theta}_i},\widetilde{\bq}_{u,i})$ with $i=1,\ldots,N$ denote its eigenpairs, where $\tilde{\theta}_i\in (-\pi,\pi]$ and $\tilde{\bq}_{u,i}\in \mathbb C^{N}$. Since $\bU$ is real, every nonreal eigenmode occurs together with its complex conjugate, whereas the real eigenvalues $+1$ and $-1$ correspond to the angles $0$ and $\pi$, respectively.


Unlike the scaling modes of $\bP$, which are naturally ranked by their eigenvalues, all eigenvalues of the orthogonal factor $\bU$ have unit modulus. Their magnitudes therefore provide no criterion for constructing a $k$-dimensional orientation representation compatible with the rank-$k$ scaling representation. To select the $k$ modes most relevant to the product $\bA_\varepsilon=\bU\bP$, we evaluate each mode using both its angular deviation from the identity and the scaling magnitude of $\bP$ along its direction. This motivates the polar-weighted orientation score $s_i=2\big|\sin(\widetilde{\theta_i}/{2})\big|\widetilde{\bq}_{u,i}^{\ast}\bP\widetilde{\bq}_{u,i}$,
where $(\cdot)^{\ast}$ denotes the conjugate transpose. The angular factor $2|\sin\big(\widetilde{\theta_i}/{2}\big)|=|e^{\mathrm i\widetilde{\theta_i}}-1|$ measures the deviation of the corresponding eigenvalue from the identity orientation, while $\widetilde{\bq}_{u,i}^{\ast}\bP\widetilde{\bq}_{u,i}$ measures the scaling magnitude associated with that mode. Because $\bP$ is real and positive semidefinite, $s_i\in \Real_{\geq 0}$, and every complex-conjugate pair receives identical scores. The score therefore provides a valid scalar ranking that is consistent with pair-preserving mode selection.
 
PolarFlow then selects a size-$k$ index set $\calI_{u}=\{i_{1},\ldots, i_{k}\}$ with the highest scores, subject to closure under complex conjugation. In particular, whenever a selected index corresponds to a nonreal eigenvalue, the index of its conjugate mode is also included. Thus, each nonreal conjugate pair contributes two indices to $\calI_{u}$. Using this common index set, the selected eigendirections are represented by a real orthonormal frame and their corresponding eigenangles are collected as
\begin{align}
    \bQ_u\in\St(N,k)  \qquad \btheta^{(k)} = (\tilde{\theta}_{i_{1}},\ldots,\tilde{\theta}_{i_{k}})^\top \in \Theta_{k} \subseteq \mathbb T^{k}.
\end{align}
Here, $\mathbb T = \Real/(2\pi\mathbb Z)\simeq S^{1}$ denotes the circle of angles modulo $2\pi$, and $\Theta_{k}\subseteq \mathbb T^{k}$ is the admissible subset induced by the conjugate structure of the retained orientation modes.

Combining the scaling and orientation components, we define the $k$-dimensional four-channel representation of $\bA_{\varepsilon}$ as $\Phi_{k}(\bA_{\epsilon}) =  \big(\bQ_p,\blambda_p^{(k)}, \bQ_u,\btheta^{(k)}\big) \in\mathcal M_k$,
where $\calM_k = \St(N,k) \times \Real_{+}^{k} \times \St(N,k) \times \Theta_k$ is the associated geometric product state space. The two Stiefel factors represent the retained scaling and orientation directions, while $\blambda_p^{(k)}$ and $\btheta^{(k)}$ contain their respective spectral coordinates.

\subsection{Coupled endpoint distributions and geometry-aware paths}
\label{sec:paths}
PolarFlow constructs its source and target states by applying the same polar–spectral map $\Phi_k$ to the corresponding matrix endpoints. For each graph size $N$, let $\bA_0\in\mathbb R^{N\times N}$ be a deterministic source template and let $ \bF_{ij}\overset{\mathrm{i.i.d.}}{\sim}\mathcal N(0,1)$. The source matrix endpoint and its four-channel representation are defined as
\begin{equation}
    \bM_0 = \bA_0+\sigma \bF,\qquad \bZ_0=\Phi_{k}(\bM_0),
    \label{eq:source}
\end{equation}
where $\sigma>0$ controls the source noise level. For a target adjacency matrix $\bA_1\sim p_{\mathrm{data}}(\cdot\mid N)$, the regularized target matrix endpoint and its representation are 
\begin{equation}
    \bM_1 = \bA_1+\bE_\varepsilon,\qquad
    \bZ_1=\Phi_{k}(\bM_1), \qquad
    \|\bE_{\varepsilon}\| \leq \varepsilon
    \label{eq:target}
\end{equation}
These constructions induce source and target distributions over the four-channel state space. Because every endpoint is obtained from a single matrix through one application of $\Phi_{k}$, its orientation and scaling channels retain the dependencies induced by the underlying polar decomposition. In particular, PolarFlow does not construct an endpoint by independently sampling the four channel-wise marginals.

To define a conditional path between a sampled source–target pair, we make the four components of each endpoint explicit, i.e, $ \bZ_{\ell} = (\bQ_{p,\ell}, \lambda_{p,\ell}^{(k)}, \bQ_{u,\ell}, \btheta_{\ell}^{(k)}) \in \calM_{k}$ with $\ell\in \{0,1\}$.
PolarFlow constructs a conditional path using a shared time $t\sim \text{Unif}([0,1])$. The orientation and scaling frames follow Stiefel geodesics:
\begin{align}
\bQ_{c,t}=\Exp_{\bQ_{c,0}}\!\left(t\Log_{\bQ_{c,0}}(\bQ_{c,1})\right),
    \qquad c\in\{u,p\},
    \label{eq:stiefel_path}  
\end{align}
where $\Exp$ and $\Log$ are the exponential and logarithmic functions between the manifold and its tangent spaces. For the orientation-angle channel, define the componentwise shortest angular displacement as
$\bdelta_\theta^{(k)}=\wrap(\btheta_1-\btheta_0)$,
where $\wrap(\cdot)$ returns the principal representative modulo $2\pi$. The corresponding circular path is $\btheta_t^{(k)} =\wrap(\btheta_0^{(k)}+t\bdelta_{\btheta}^{(k)})$, and
the spectral-value channel follows a Euclidean linear path $\blambda_{p,t}^{(k)} = (1-t)\blambda_{p,0}^{(k)} + t \blambda_{p,1}^{(k)}$.
Since both endpoint vectors are nonnegative, the interpolated spectral values remain in $\Real_{+}^{k}$. The four channels together define the synchronized intermediate state
\begin{equation}
\bZ_{t} = (\bQ_{p,t}, \blambda_{p,t}^{(k)}, \bQ_{u,t}, \btheta_{t}^{(k)}) \in \calM_{k}.
\label{eq.conditional_path} 
\end{equation}
The corresponding pathwise velocities provide the conditional targets for geometric flow matching. Along the selected geodesic branches, the corresponding channel-wise velocities are
\begin{equation}
    \dot{\bQ}_{c,t}=\frac{\ud \bQ_{c,t}}{\ud t} \qquad \text{with}  \qquad c \in \{u,p\}, 
\end{equation}
while the angle and spectral-value velocities are
\begin{align*}
    \dot{\btheta}_{t}^{(k)} = \frac{\ud \btheta_{t}^{(k)}}{\ud t} = \bdelta_\theta^{(k)}, \quad \text{and} \quad \dot{\blambda}_{p,t}^{(k)}=\frac{\ud \blambda_{p,t}^{(k)}}{\ud t}=\blambda_{p,1}^{(k)} - \blambda_{p,0}^{(k)}
\end{align*}
Collectively, $\bv_{t}(\bZ_{t} |\bZ_{0}, \bZ_{1}) = \big(\dot{\bQ}_{p,t}, \dot{\blambda}_{p,t}^{(k)} ,\dot{\bQ}_{u,t},\dot{\btheta}_{t}^{(k)}\big)$ defines the geometry-aware conditional velocity target. Using the same time $t$ for all four channels ensures that the components associated with a single endpoint pair evolve synchronously.



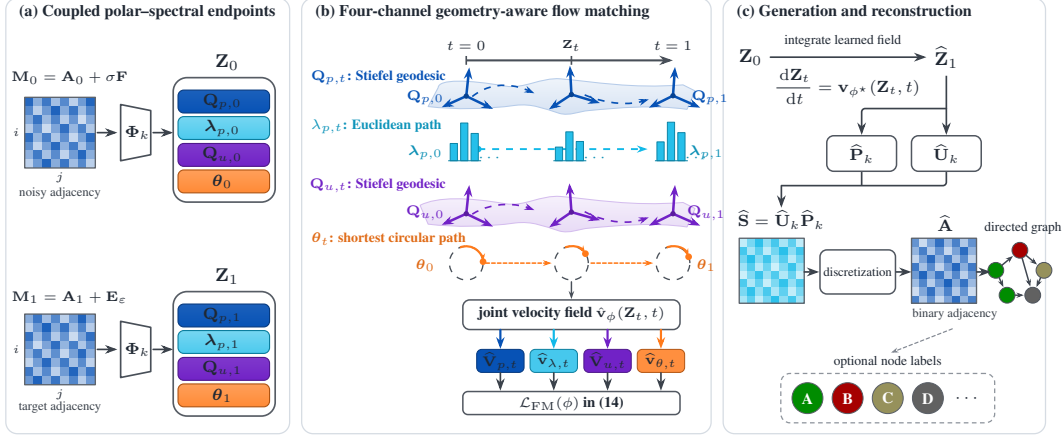
\begin{figure*}[t]
\centering

\begingroup

\definecolor{PFQp}{HTML}{0752B5}
\definecolor{PFLambda}{HTML}{12B9E8}
\definecolor{PFQu}{HTML}{7122C6}
\definecolor{PFTheta}{HTML}{FF7A1A}
\definecolor{PFInk}{HTML}{202833}
\definecolor{PFMuted}{HTML}{515A67}
\definecolor{PFBorder}{HTML}{B9C2CE}
\definecolor{PFSoft}{HTML}{F7F8FA}

\tikzset{
  pf panel/.style={
    rounded corners=5pt,
    draw=PFBorder!58,
    line width=0.55pt
  },
  pf title/.style={
    anchor=west,
    font=\bfseries\fontsize{8.4}{9.3}\selectfont,
    text=PFInk
  },
  pf flow/.style={
    -{Stealth[length=1.9mm,width=1.35mm]},
    draw=PFInk!88,
    line width=0.72pt,
    line cap=round,
    line join=round
  },
  pf color flow/.style={
    -{Stealth[length=1.8mm,width=1.30mm]},
    line width=0.95pt,
    line cap=round,
    line join=round
  },
  pf box/.style={
    rounded corners=3pt,
    fill=white,
    draw=PFInk!78,
    line width=0.65pt,
    align=center,
    text=PFInk,
    inner sep=3pt
  },
  pf lane head/.style={
    anchor=west,
    font=\bfseries\fontsize{6.2}{6.9}\selectfont,
    text=PFInk
  },
  pf lane note/.style={
    anchor=west,
    align=left,
    font=\fontsize{5.0}{5.6}\selectfont,
    text=PFInk!86
  }
}


\newcommand{\pfTargetHeatmap}[4]{%
  \begin{scope}[shift={(#1,#2)}]
    \foreach \i in {0,...,7}{
      \foreach \j in {0,...,7}{
        \pgfmathtruncatemacro{\shadevalue}{20+mod(#4+19*\i+31*\j,72)}
        \fill[#3!\shadevalue!white]
          ({0.22*\j},{0.22*\i}) rectangle ++(0.22,0.22);
      }
    }
    \draw[PFInk!72,line width=0.58pt] (0,0) rectangle (1.76,1.76);
  \end{scope}%
}

\newcommand{\pfTargetPhi}[2]{%
  \begin{scope}[shift={(#1,#2)}]
    \path[fill=white,draw=PFInk!82,line width=0.72pt]
      (-0.41,-0.70)--(0.34,-0.55)--(0.34,0.55)--(-0.41,0.70)--cycle;
    \node[font=\bfseries\fontsize{7.0}{7.5}\selectfont,text=PFInk]
      at (0.01,0) {$\bPhi_k$};
  \end{scope}%
}

\newcommand{\pfTargetChannels}[3]{%
  \begin{scope}[shift={(#1,#2)}]


    \path[
      rounded corners=6pt,
      fill=white,
      draw=PFInk!76,
      line width=0.72pt
    ]
      (-1.3,-1.53) rectangle (1.24,1.57);

    \node[
      font=\bfseries\fontsize{8.2}{9.0}\selectfont,
      text=PFInk
    ] at (0.04,1.95)
      {$\bZ_{#3}$};

    \path[
      rounded corners=2.5pt,
      top color=PFQp!92,
      bottom color=PFQp,
      draw=PFQp!70!black,
      line width=0.48pt
    ]
      (-1.18,0.69) rectangle (1.10,1.26);

    \node[
      font=\fontsize{7.0}{7.5}\selectfont,
      text=PFInk
    ] at (-0.08,0.975)
      {$\bQ_{p,#3}$};

    \path[
      rounded corners=2.5pt,
      top color=PFLambda!58,
      bottom color=PFLambda!76,
      draw=PFLambda!68!black,
      line width=0.48pt
    ]
      (-1.18,0.03) rectangle (1.10,0.60);

    \node[
      font=\fontsize{7.0}{7.5}\selectfont,
      text=PFInk
    ] at (-0.08,0.315)
      {$\blambda_{p,#3}$};

    \path[
      rounded corners=2.5pt,
      top color=PFQu!90,
      bottom color=PFQu,
      draw=PFQu!70!black,
      line width=0.48pt
    ]
      (-1.18,-0.64) rectangle (1.10,-0.07);

    \node[
      font=\fontsize{7.0}{7.5}\selectfont,
      text=PFInk
    ] at (-0.08,-0.355)
      {$\bQ_{u,#3}$};

    \path[
      rounded corners=2.5pt,
      top color=PFTheta!72,
      bottom color=PFTheta!88,
      draw=PFTheta!68!black,
      line width=0.48pt
    ]
      (-1.18,-1.30) rectangle (1.10,-0.73);

    \node[
      font=\fontsize{7.0}{7.5}\selectfont,
      text=PFInk
    ] at (-0.08,-1.015)
      {$\btheta_{#3}$};

  \end{scope}%
}

\newcommand{\pfTargetFrame}[4]{%
  \begin{scope}[shift={(#1,#2)},rotate=#4]
    \draw[-{Stealth[length=1.8mm,width=1.3mm]},#3,line width=1.10pt]
      (0,0)--(0,0.72);
    \draw[-{Stealth[length=1.8mm,width=1.3mm]},#3,line width=1.10pt]
      (0,0)--(0.70,-0.25);
    \draw[-{Stealth[length=1.8mm,width=1.3mm]},#3,line width=1.10pt]
      (0,0)--(-0.58,-0.34);
    \fill[#3!80!black] (0,0) circle (1.15pt);
  \end{scope}%
}

\newcommand{\pfTargetSpectrum}[5]{%
  \begin{scope}[shift={(#1,#2)}]
    \foreach \x/\h in {-0.40/#3,-0.14/#4,0.12/#5}{
      \filldraw[fill=PFLambda!92,draw=PFLambda!72!black,line width=0.30pt]
        (\x,0) rectangle ++(0.18,\h);
    }
    \draw[PFLambda!72!black,line width=0.45pt] (-0.45,0)--(0.37,0);
    \node[font=\tiny,text=PFLambda!78!black] at (0.61,0.08) {$\cdots$};
  \end{scope}%
}


\newcommand{\pfTargetCircle}[4][0.49]{%
  \begin{scope}[shift={(#2,#3)}]

    \draw[
      PFInk!82,
      dashed,
      line width=0.58pt
    ]
      (0,0) circle (#1);

    \draw[
      -{Stealth[length=1.25mm,width=0.85mm]},
      PFTheta,
      line width=0.95pt
    ]
      (112:#1)
      arc[
        start angle=112,
        end angle=#4,
        radius=#1
      ];

    \fill[PFTheta]
      (#4:#1) circle (1.45pt);

  \end{scope}%
}

\newcommand{\pfTargetCircularPath}[2]{%
  \begin{scope}[shift={(0,#1)}]

    \pgfmathsetmacro{\thetaConnectorOffset}{#2+0.07}

    \draw[
      -{Stealth[length=1.35mm,width=0.90mm]},
      PFTheta,
      dash pattern=on 1.5pt off 1.0pt,
      line cap=round,
      line width=0.68pt
    ]
      ({\xTzero+\thetaConnectorOffset},5.14)
      --
      ({\xTmid-\thetaConnectorOffset},5.14);

    \draw[
      -{Stealth[length=1.35mm,width=0.90mm]},
      PFTheta,
      dash pattern=on 1.5pt off 1.0pt,
      line cap=round,
      line width=0.68pt
    ]
      ({\xTmid+\thetaConnectorOffset},5.14)
      --
      ({\xTone-\thetaConnectorOffset},5.14);

    \pfTargetCircle[#2]{\xTzero}{5.14}{5}
    \pfTargetCircle[#2]{\xTmid}{5.14}{34}
    \pfTargetCircle[#2]{\xTone}{5.14}{22}

    \node[
      font=\fontsize{6.5}{7.0}\selectfont,
      text=PFTheta!88!black
    ] at (17.48,5.09)
      {$\bm\theta_1$};

    \node[
      font=\fontsize{6.5}{7.0}\selectfont,
      text=PFTheta!88!black
    ] at (10.48,5.09)
      {$\bm\theta_0$};

  \end{scope}%
}

\tikzset{
  pf graph arrow/.style={
    -{Stealth[length=1.15mm,width=0.82mm]},
    draw=PFInk!85,
    line width=0.65pt,
    line cap=round,
    shorten <=1.15mm,
    shorten >=1.35mm
  }
}

\newcommand{\pfTargetGraph}[2]{%
  \begin{scope}[shift={(#1,#2)}]

    \coordinate (g1) at (-0.60, 0.10);
    \coordinate (g2) at ( 0.00, 0.60);
    \coordinate (g3) at ( 0.60, 0.10);
    \coordinate (g4) at ( 0.32,-0.52);
    \coordinate (g5) at (-0.32,-0.52);

    \draw[pf graph arrow] (g1)--(g2);
    \draw[pf graph arrow] (g2)--(g3);
    \draw[pf graph arrow] (g3)--(g4);
    \draw[pf graph arrow] (g1)--(g5);
    \draw[pf graph arrow] (g5)--(g4);
    \draw[pf graph arrow] (g2)--(g4);


\node[
  circle,
  fill=green!55!black,
  draw=PFInk!80,
  minimum size=0.24cm,
  inner sep=0
] at (g1) {};

\node[
  circle,
  fill=red!65!black,
  draw=PFInk!80,
  minimum size=0.24cm,
  inner sep=0
] at (g2) {};

\node[
  circle,
  fill=yellow!50!black,
  draw=PFInk!80,
  minimum size=0.24cm,
  inner sep=0
] at (g3) {};

\node[
  circle,
  fill=black!62,
  draw=PFInk!80,
  minimum size=0.24cm,
  inner sep=0
] at (g4) {};

\node[
  circle,
  fill=green!55!black,
  draw=PFInk!80,
  minimum size=0.24cm,
  inner sep=0
] at (g5) {};
  \end{scope}%
}
\resizebox{1.02\textwidth}{!}{
\begin{tikzpicture}[x=0.61cm,y=0.61cm,line join=round]

\def\yPanelBottom{1.1} 
\def\yPanelTop{11.95}   

\path[pf panel]
  (0.03,\yPanelBottom) rectangle (7.10,\yPanelTop);

\path[pf panel]
  (7.40,\yPanelBottom) rectangle (17.60,\yPanelTop);

\path[pf panel]
  (17.90,\yPanelBottom) rectangle (25.93,\yPanelTop);

\node[
  pf title,
  font=\bfseries\fontsize{7}{8}\selectfont
] at (0.16,11.60)
  {(a) Coupled polar--spectral endpoints};

\node[
  pf title,
  font=\bfseries\fontsize{7}{8}\selectfont
] at (7.54,11.60)
  {(b) Four-channel geometry-aware flow matching};

\node[
  pf title,
  font=\bfseries\fontsize{7}{8}\selectfont
] at (18.04,11.60)
  {(c) Generation and reconstruction};


\node[font=\fontsize{5.9}{6.4}\selectfont,text=PFInk] at (1.62,9.96)
  {$\bM_0=\bA_0+\sigma\bF$};
\pfTargetHeatmap{0.50}{7.72}{PFQp}{7}
\node[font=\fontsize{5.5}{5.8}\selectfont,text=PFInk] at (0.29,8.61) {$i$};
\node[font=\fontsize{5.5}{5.8}\selectfont,text=PFInk] at (1.38,7.49) {$j$};
\node[font=\fontsize{5.5}{6.0}\selectfont,text=PFInk!88]
  at (1.38,7.1) {noisy adjacency};
\draw[pf flow] (2.31,8.60)--(2.81,8.60); 
\pfTargetPhi{3.30}{8.60}
\draw[pf flow] (3.7,8.60)--(4.2,8.60); 
\pfTargetChannels{5.5}{8.4}{0}

\begin{scope}[shift={(0,-0.73)}]

\node[font=\fontsize{5.9}{6.4}\selectfont,text=PFInk]
  at (1.62,5.23)
  {$\bM_1=\bA_1+\bE_{\varepsilon}$};

\pfTargetHeatmap{0.50}{3.05}{PFQp}{39}

\node[font=\fontsize{5.4}{5.8}\selectfont,text=PFInk]
  at (0.29,3.94) {$i$};

\node[font=\fontsize{5.4}{5.8}\selectfont,text=PFInk]
  at (1.38,2.82) {$j$};

\node[font=\fontsize{5.5}{6.0}\selectfont,text=PFInk!88]
  at (1.38,2.50) {target adjacency};

\draw[pf flow] (2.31,3.93)--(2.81,3.93);

\pfTargetPhi{3.30}{3.93}

\draw[pf flow] (3.70,3.93)--(4.20,3.93);

\pfTargetChannels{5.50}{3.80}{1}

\end{scope}

\def\xTzero{11.50}
\def\xTmid{14.12}
\def\xTone{16.65}

\draw[
  -{Stealth[length=2.1mm,width=1.45mm]},
  PFInk!82,
  line width=0.90pt
]
  (\xTzero,10.40)--(\xTone,10.40);

\draw[PFInk!82,line width=0.70pt]
  (\xTzero,10.29)--(\xTzero,10.51);
\draw[PFInk!82,line width=0.70pt]
  (\xTmid,10.29)--(\xTmid,10.51);

\node[
  font=\fontsize{6.0}{6.5}\selectfont,
  text=PFInk
] at (\xTzero,10.76)
  {$t=0$};

\node[
  font=\fontsize{5.5}{6.4}\selectfont,
  text=PFInk
] at (\xTmid,10.76)
  {$\bZ_t$};

\node[
  font=\fontsize{6.0}{6.5}\selectfont,
  text=PFInk
] at (\xTone,10.76)
  {$t=1$};

\node[
  pf lane head,
  anchor=east,
  text=PFQp,
  font=\bfseries\fontsize{6}{6.8}\selectfont,
  inner sep=0pt
] (qp-head) at (11,9.98)
  {$\bQ_{p,t}$: Stiefel geodesic};


\path[
  top color=PFQp!6,
  bottom color=PFQp!17,
  draw=PFQp!30,
  line width=0.45pt,
  line join=round
]
  (10.8,9.72)
  .. controls (10.81,9.94) and (11.55,9.92)
     .. (12.47,9.76)
  .. controls (12.95,9.68) and (13.35,9.98)
     .. (13.85,9.94)
  .. controls (14.70,9.87) and (15.58,9.69)
     .. (16.90,9.79)
  -- (17.20,9.64)
  -- (17.00,9.18)
  .. controls (16.15,9.30) and (15.30,9.22)
     .. (14.55,9.15)
  .. controls (14.30,9.13) and (14.08,9.01)
     .. (13.82,9.08)
  .. controls (13.43,9.20) and (13.15,9.33)
     .. (12.72,9.28)
  .. controls (12.00,9.18) and (11.25,9.17)
     .. (10.55,9.25)
  -- cycle;

\draw[
  -{Stealth[length=1.55mm,width=1.05mm]},
  PFQp,
  dashed,
  line width=0.72pt
]
  (11.65,9.50)
  .. controls (12.20,9.94) and (12.80,9.76)
  .. (13.20,9.58);

\draw[
  -{Stealth[length=1.55mm,width=1.05mm]},
  PFQp,
  dashed,
  line width=0.72pt
]
  (14.45,9.54)
  .. controls (15.05,9.29) and (15.60,9.29)
  .. (16.05,9.43);

\pfTargetFrame{\xTzero}{9.51}{PFQp}{-6}
\pfTargetFrame{\xTmid}{9.55}{PFQp}{3}
\pfTargetFrame{\xTone}{9.52}{PFQp}{-4}

\node[
  font=\fontsize{6.5}{7.0}\selectfont,
  text=PFQp
] at (17.58,9.48)
  {$\bQ_{p,1}$};

\node[
  font=\fontsize{6.5}{7.0}\selectfont,
  text=PFQp
] at (10.48,9.48)
  {$\bQ_{p,0}$};


\node[
  pf lane head,
  anchor=east,
  text=PFLambda!80!black,
  font=\bfseries\fontsize{6}{6.8}\selectfont,
  inner sep=0pt
] (qp-head) at (10.9,8.75)
  {$\lambda_{p,t}$: Euclidean path};


\draw[
  -{Stealth[length=1.65mm,width=1.10mm]},
  PFLambda,
  dashed,
  line width=0.82pt
]
  (11.80,8.19)--(16.02,8.19);

\fill[PFLambda] (11.80,8.19) circle (1.40pt);

\pfTargetSpectrum{\xTzero}{7.9}{0.43}{0.94}{0.68} 
\pfTargetSpectrum{\xTmid}{7.9}{0.34}{0.72}{0.48}
\pfTargetSpectrum{\xTone}{7.9}{0.47}{0.96}{0.70}

\node[
  font=\fontsize{6.5}{7.0}\selectfont,
  text=PFLambda!78!black
] at (17.48,8.11)
  {$\blambda_{p,1}$};

\node[
  font=\fontsize{6.5}{7.0}\selectfont,
  text=PFLambda!78!black
] at (10.48,8.11)
  {$\blambda_{p,0}$};

  
\node[
  pf lane head,
  anchor=east,
  text=PFQu,
  font=\bfseries\fontsize{6}{6.8}\selectfont,
  inner sep=0pt
] (qp-head) at (11,7.28)
  {$\bQ_{u,t}$: Stiefel geodesic};

\begin{scope}[shift={(0,0.15)}]  

\path[
  top color=PFQu!6,
  bottom color=PFQu!17,
  draw=PFQu!30,
  line width=0.45pt,
  line join=round
]
  (10.35,6.64)
  .. controls (10.78,6.86) and (11.55,6.84)
     .. (12.47,6.68)
  .. controls (12.95,6.60) and (13.35,6.90)
     .. (13.85,6.86)
  .. controls (14.70,6.79) and (15.58,6.61)
     .. (16.90,6.71)
  -- (17.20,6.56)
  -- (17.00,6.10)
  .. controls (16.15,6.22) and (15.30,6.14)
     .. (14.55,6.07)
  .. controls (14.30,6.05) and (14.08,5.93)
     .. (13.82,6.00)
  .. controls (13.43,6.12) and (13.15,6.25)
     .. (12.72,6.20)
  .. controls (12.00,6.10) and (11.25,6.09)
     .. (10.55,6.17)
  -- cycle;

\draw[
  -{Stealth[length=1.55mm,width=1.05mm]},
  PFQu,
  dashed,
  line width=0.72pt
]
  (11.65,6.42)
  .. controls (12.20,6.76) and (12.80,6.67)
  .. (13.20,6.52);

\draw[
  -{Stealth[length=1.55mm,width=1.05mm]},
  PFQu,
  dashed,
  line width=0.72pt
]
  (14.45,6.48)
  .. controls (15.05,6.17) and (15.60,6.18)
  .. (16.05,6.34);

\pfTargetFrame{\xTzero}{6.43}{PFQu}{-6}
\pfTargetFrame{\xTmid}{6.47}{PFQu}{4}
\pfTargetFrame{\xTone}{6.43}{PFQu}{-4}

\end{scope}

\node[
  font=\fontsize{6.5}{7.0}\selectfont,
  text=PFQu
] at (17.48,6.60)
  {$\bQ_{u,1}$};

\node[
  font=\fontsize{6.5}{7.0}\selectfont,
  text=PFQu
] at (10.48,6.60)
  {$\bQ_{u,0}$};


\node[
  pf lane head,
  anchor=east,
  text=PFTheta!88!black,
  font=\bfseries\fontsize{6}{6.8}\selectfont,
  inner sep=0pt
] (qp-head) at (11.5,5.95)
  {$\bm\theta_t$: shortest circular path};





\pfTargetCircularPath{0.21}{0.42}


\begin{scope}[yshift=1mm]
\def\xVelocityCenter{\xTmid}
\node[
  pf box,
  minimum width=3.30cm,
  minimum height=0.45cm,
  inner xsep=2pt,
  inner ysep=2pt,
  font=\bfseries\fontsize{6.5}{6.9}\selectfont
] (joint) at (\xVelocityCenter,3.9)
  {joint velocity field $\hat{\bv}_{\phi}(\bZ_t,t)$};

\draw[
  -{Stealth[length=1.25mm,width=0.85mm]},
  draw=PFInk!82,
  line width=0.62pt,
  line cap=round,
  shorten <=0pt,
  shorten >=0pt,
]
  (\xVelocityCenter,4.72)--(joint.north);

\tikzset{
  pf velocity head/.style={
    rounded corners=2.2pt,
    minimum width=0.72cm,
    minimum height=0.35cm,
    inner xsep=1.6pt,
    inner ysep=0.6pt,
    font=\fontsize{6.5}{7.0}\selectfont
  }
}

\def\yVelocityHead{2.75}

\node[
  pf velocity head,
  fill=PFQp,
  draw=PFQp!70!black,
  text=PFInk
] (vp) at (12.35,\yVelocityHead)
  {$\widehat{\mathbf V}_{p,t}$};

\node[
  pf velocity head,
  fill=PFLambda!78,
  draw=PFLambda!68!black,
  text=PFInk
] (vl) at (13.68,\yVelocityHead)
  {$\widehat{\mathbf v}_{\lambda,t}$};

\node[
  pf velocity head,
  fill=PFQu,
  draw=PFQu!70!black,
  text=PFInk
] (vu) at (15.02,\yVelocityHead)
  {$\widehat{\mathbf V}_{u,t}$};

\node[
  pf velocity head,
  fill=PFTheta!84,
  draw=PFTheta!70!black,
  text=PFInk
] (vt) at (16.35,\yVelocityHead)
  {$\widehat{\mathbf v}_{\theta,t}$};

\draw[
  pf color flow,
  draw=PFQp,
  shorten <=-0.2pt,
  shorten >=0pt
]
  (joint.south -| vp.north) -- (vp.north);

\draw[
  pf color flow,
  draw=PFLambda,
  shorten <=-0.2pt,
  shorten >=0pt
]
  (joint.south -| vl.north) -- (vl.north);

\draw[
  pf color flow,
  draw=PFQu,
  shorten <=-0.2pt,
  shorten >=0pt
]
  (joint.south -| vu.north) -- (vu.north);

\draw[
  pf color flow,
  draw=PFTheta,
  shorten <=-0.2pt,
  shorten >=0pt
]
  (joint.south -| vt.north) -- (vt.north);

\node[
  pf box,
  minimum width=3.30cm,
  minimum height=0.35cm,
  inner xsep=2pt,
  inner ysep=0.4pt,
  rounded corners=2.2pt,
  font=\bfseries\fontsize{6.4}{6.9}\selectfont
] (loss) at (\xVelocityCenter,1.7)
  {$\mathcal L_{\mathrm{FM}}(\phi)$ in \cref{eq:fm_loss}};

\draw[pf flow] (vp.south)--(vp.south |- loss.north);
\draw[pf flow] (vl.south)--(vl.south |- loss.north);
\draw[pf flow] (vu.south)--(vu.south |- loss.north);
\draw[pf flow] (vt.south)--(vt.south |- loss.north);

\end{scope}
\node[font=\bfseries\fontsize{8.2}{9.0}\selectfont,text=PFInk]
  (zstart) at (18.58,10.48) {$\mathbf Z_0$};
\node[font=\bfseries\fontsize{8.2}{9.0}\selectfont,text=PFInk]
  (zend) at (23.46,10.48) {$\widehat{\mathbf Z}_1$};
\draw[pf flow] (zstart.east)--(zend.west);
\node[font=\fontsize{5.8}{6.3}\selectfont,text=PFInk]
  at (20.88,10.89) {integrate learned field};
\node[font=\fontsize{6.8}{7.5}\selectfont,text=PFInk,fill=white,inner sep=1pt]
  at (20.98,9.75)
  {$\displaystyle \frac{\ud \bZ_{t}}{\ud t} = \bv_{\phi^{\star}}(\bZ_{t},t) $};

\node[
  pf box,
  minimum width=1.06cm,
  minimum height=0.54cm,
  inner sep=1.2pt
] (pk) at (21.35,8.08)
  {\bfseries\fontsize{7.5}{8.1}\selectfont
   $\widehat{\bP}_k$};

\node[
  pf box,
  minimum width=1.06cm,
  minimum height=0.54cm,
  inner sep=1.2pt
] (uk) at (23.46,8.08)
  {\bfseries\fontsize{7.5}{8.1}\selectfont
   $\widehat{\bU}_k$};

\coordinate (decodebranch) at (23.46,9.20);

\draw[
  PFInk!88,
  line width=0.72pt
]
  (zend.south)--(decodebranch);

\draw[pf flow]
  (decodebranch)--(uk.north);

\draw[
  pf flow,
  rounded corners=1.5pt
]
  (decodebranch) -| (pk.north);


\coordinate (scoremerge) at (19.35,7.28);

\coordinate (pkdrop) at (pk.south |- scoremerge);
\coordinate (ukdrop) at (uk.south |- scoremerge);

\draw[
  PFInk!88,
  line width=0.72pt,
  rounded corners=1.5pt
]
  (uk.south)
  -- (ukdrop)
  -- (scoremerge);

\draw[
  PFInk!88,
  line width=0.72pt
]
  (pk.south)--(pkdrop);

\node[
  font=\fontsize{7.1}{7.7}\selectfont,
  text=PFInk,
  fill=white,
  inner sep=1pt
] (scoreeq) at (19.35,6.43)
  {$\widehat{\mathbf S}
    =
    \widehat{\mathbf U}_k
    \widehat{\mathbf P}_k$};

\draw[pf flow]
  (scoremerge)--(scoreeq.north);



\begin{scope}[shift={(18.30,4.35)},scale=0.92]
  \pfTargetHeatmap{0}{0}{PFLambda}{17}

  \coordinate (scoremap-west) at (0,0.88);
  \coordinate (scoremap-east) at (1.76,0.88);
\end{scope}


\begin{scope}[xshift=-2mm]

\node[
  pf box,
  minimum width=0.95cm,
  minimum height=0.67cm,
  inner sep=1.5pt,
  font=\fontsize{5.5}{6.1}\selectfont
] (topm) at (21.58,5.13)
  {discretization};

\node[
  font=\fontsize{7.0}{7.6}\selectfont,
  text=PFInk
] at (23.76,6.35)
  {$\widehat{\mathbf A}$};

\begin{scope}[shift={(22.92,4.35)},scale=0.92]
  \pfTargetHeatmap{0}{0}{PFQp}{53}

  \coordinate (binary-west) at (0,0.88);
  \coordinate (binary-east) at (1.76,0.88);
\end{scope}

\node[
  font=\fontsize{5.4}{5.9}\selectfont,
  text=PFInk!88
] (binarylabel) at (24.00,4.08)
  {binary adjacency};

\node[
  font=\fontsize{5.8}{6.3}\selectfont,
  text=PFInk,
  align=center
] (graphlabel) at (25.70,6.22)
  {directed graph};

\pfTargetGraph{25.62}{5.07}

\coordinate (graph-input) at (24.90,5.13);

\end{scope}


\draw[pf flow]
  (scoremap-east |- topm.west)
  -- (topm.west);

\draw[pf flow]
  (topm.east)
  -- (binary-west |- topm.east);

\draw[pf flow]
  (binary-east |- graph-input)
  -- (graph-input);


\def\yLabelBottom{1.45}

\pgfmathsetmacro{\yLabelCenter}{\yLabelBottom+0.536}

\pgfmathsetmacro{\yLabelTitle}{\yLabelCenter+0.90}

\def\xLabelCenter{22.00}

\node[
  draw=PFMuted!65,
  dashed,
  dash pattern=on 2.0pt off 1.5pt,
  rounded corners=4pt,
  line width=0.55pt,
  minimum width=3.25cm,
  minimum height=0.75cm,
  inner sep=0pt
] (labelcontainer) at (\xLabelCenter,\yLabelCenter) {};

\node[
  font=\fontsize{5.8}{6.3}\selectfont,
  text=PFInk,
  fill=white,
  inner xsep=2.5pt,
  inner ysep=0.7pt
] (labeltitle) at (\xLabelCenter,\yLabelTitle)
  {optional node labels};

\draw[
  -{Stealth[length=1.15mm,width=0.78mm]},
  draw=PFMuted!72,
  dash pattern=on 1.5pt off 1.15pt,
  line width=0.52pt,
  line cap=round,
  shorten <=1.2pt,
  shorten >=1.5pt
]
  (binarylabel.south)
  --
  (labeltitle.north);

\node[
  circle,
  fill=green!55!black,
  draw=PFInk!75,
  text=white,
  minimum size=0.47cm,
  inner sep=0pt,
  font=\bfseries\fontsize{6.8}{7.2}\selectfont
] at (20,\yLabelCenter) {A};

\node[
  circle,
  fill=red!65!black,
  draw=PFInk!75,
  text=white,
  minimum size=0.47cm,
  inner sep=0pt,
  font=\bfseries\fontsize{6.8}{7.2}\selectfont
] at (21,\yLabelCenter) {B};

\node[
  circle,
  fill=yellow!50!black,
  draw=PFInk!75,
  text=white,
  minimum size=0.47cm,
  inner sep=0pt,
  font=\bfseries\fontsize{6.8}{7.2}\selectfont
] at (22,\yLabelCenter) {C};

\node[
  circle,
  fill=black!62,
  draw=PFInk!75,
  text=white,
  minimum size=0.47cm,
  inner sep=0pt,
  font=\bfseries\fontsize{6.8}{7.2}\selectfont
] at (23,\yLabelCenter) {D};

\node[
  font=\fontsize{7.0}{7.5}\selectfont,
  text=PFInk
] at (24,\yLabelCenter)
  {$\cdots$};

\end{tikzpicture}
}
\endgroup
\vspace{-7mm}
\caption{\small Overview of the proposed PolarFlow framework.}
\label{fig:polarflow_overview}
\end{figure*}

\subsection{Joint geometry-aware velocity matching}
\label{sec:flow_objective}
PolarFlow models the joint evolution of the four polar–spectral channels using a single time-dependent velocity field. For a sampled source–target pair, the geometry-aware path at time $t$ provides the conditional velocity target $ \bv_{t}(\bZ_{t} |\bZ_{0}, \bZ_{1}) = \big(\dot{\bQ}_{p,t}, \dot{\blambda}_{p,t}^{(k)} ,\dot{\bQ}_{u,t},\dot{\btheta}_{t}^{(k)}\big)$.
A shared network receives the complete intermediate state $\bZ_t$ and time $t$ and jointly predicts the four velocity components:
\begin{equation}
    \hat{\bv}_\phi(\bZ_t,t)=
    (\widehat \bV_{p,t},\widehat{\bv}_{\lambda,t}, \widehat \bV_{u,t}, \widehat{\bv}_{\theta,t})
    \in \mathcal{T}_{\bZ_t}\mathcal M_k.
    \label{eq:joint_velocity}
\end{equation}
Here $\widehat \bV_{c,t}$ denotes the velocity prediction associated with the Stiefel component $\bQ_{c,t}$ after projection onto $\mathcal{T}_{\bQ_{c,t}}\St(N,k)$ for $c\in \{p,u\}$.
The associated product tangent space is
\begin{equation}
\mathcal{T}_{\bZ_t}\mathcal M_k=
\mathcal{T}_{\bQ_{p,t}}\St(N,k)\times \mathbb R^k \times \mathcal{T}_{\bQ_{u,t}}\St(N,k)  \times \mathcal{T}_{\btheta_t^{(k)}}\Theta_{k}.
\end{equation}
All four velocity components are predicted from the full state $\bZ_{t}$, rather than from their corresponding channels in isolation. Thus, each prediction may depend on both the scaling and orientation. 

Let $e_{\phi}(\bA_{\epsilon},t)$ be the \emph{velocity approximation error} $\bv_{t}(\bZ_{t} |\bZ_{0}, \bZ_{1})-\hat{\bv}_\phi(\bZ_t,t)$. The network is trained by matching minimizing the expected approximation error for velocity matching:
\begin{align}
\mathcal L_{\mathrm{FM}}(\phi)
& =\mathbb{E}\norm{e_{\phi}(\bA_{\epsilon},t)}^2 \\& = \mathbb E\big[
\lVert\widehat \bV_{p,t}-\dot \bQ_{p,t}\rVert_F^2
 +\lVert\widehat{\mathbf v}_{\lambda,t}
          -\dot{\blambda}_{p,t}^{(k)}\rVert_2^2
 + \lVert\widehat \bV_{u,t}-\dot \bQ_{u,t}\rVert_F^2
 +\lVert\widehat{\bv}_{\theta,t}
          -\dot{\btheta}_t^{(k)}\rVert_2^2 \big]
\label{eq:fm_loss}
\end{align}
where the expectation is taken over the source–target endpoint sampling procedure and $t\sim \text{Unif}([0,1])$. This is a direct velocity-prediction objective. The endpoints $\bZ_{0}$ and $\bZ_{1}$ are used only to construct $\bZ_{t}$ and its conditional velocity target, while the network receives only $(Z_{t},t)$. Unlike discrete graph flow matching, which parameterizes categorical transition dynamics over categorical node-pair variables, PolarFlow learns a continuous vector field over the joint polar–spectral state, thereby modeling the coordinated evolution of global scaling and orientation information.

\paragraph{Theoretical analysis}

We have the following result regarding the stability of PolarFlow targets. See \cref{sec:pf} for the proof. 

\begin{lightgrey}
\begin{Theorem}[Stability of PolarFlow]\label{thm:stability}
Suppose $\bA_{\varepsilon}$ is regular (formally defined in \cref{sec:pf}) and that, for a fixed source endpoint, the approximation error
$e_{\phi}(\bA_{\epsilon}',t)$
is locally Lipschitz with respect to $\bA'_\varepsilon$ in a neighborhood of $\bA_\varepsilon$, uniformly over $t\in[0,1]$.
Then there exist constants $r>0$ and $C>0$ such that, for every $\widetilde{\bA}_{\varepsilon}$ satisfying $ \|\widetilde{\bA}_\varepsilon-\bA_\varepsilon\|_2 < r$, the following holds
\begin{equation}
    \sup_{t\in[0,1]}
    \left\|
        \widehat{\mathbf v}_\phi\left(\bZ_t(\widetilde{\bA}_\varepsilon),t
        \right)
        -
        \widehat{\mathbf v}_\phi\left(
          \bZ_t(\bA_\varepsilon),t
        \right)
    \right\|
    \leq
    C
    \left\|
        \widetilde{\bA}_\varepsilon-\bA_\varepsilon
    \right\|_2.
    \label{eq:learned_velocity_stability}
\end{equation}
\end{Theorem}
\end{lightgrey}

\Cref{thm:stability} establishes the local perturbation stability of PolarFlow at a spectrally and geometrically regular endpoint: small input perturbations induce proportionally small changes in the conditional paths and velocity targets. This stability extends to the learned velocity field, which is preferred to avoid pathological behavior. It also provides theoretical support for endpoint regularization as a mechanism for stabilizing flow-matching targets. This interpretation is consistent with the generation equality gains obtained under suitable regularization strengths in \cref{sec:regularization_ablation}.

\subsection{Structure reconstruction and discretization}
\label{sec:reconstruction}
For generation, PolarFlow first samples a source matrix $\bM_{0}$ and computes its polar–spectral representation $\bZ_{0}=\Phi_{k}(\bM_{0})$ in \cref{eq:source}. Start from $\bZ_{0}$, PolarFlow integrates the ODE induced by the learned velocity field $\bv_{\phi^{\star}}$ from $t=0$ to $t=1$, i.e., $\ud \bZ_{t}/\ud t = \bv_{\phi^{\star}}(\bZ_{t},t)$ with $\bZ_{t=0} = \bZ_{0}$,
where $\phi^{\star}$ denotes the learned parameters. This produces the terminal polar–spectral state $\widehat{\bZ}_{1}= (\widehat{\bQ}_{p},\widehat{\blambda}_p^{(k)}, \widehat{\bQ}_u,\widehat{\boldsymbol\theta}^{(k)}) \in \calM_{k}$.
The terminal state is first decoded into a continuous matrix representation.
The scaling and orientation components are reconstructed respectively as
\begin{align*}
    \widehat{\bP}_{k} = \widehat{\bQ}_p \diag(\widehat{\blambda}_p^{(k)})\widehat{\bQ}_p^* \quad \text{and} \quad \widehat{\bU}_{k}=\widehat{\bQ}_u\diag(e^{\mathrm i{\btheta}^{(k)}})\widehat{\bQ}_u^*.
\end{align*}
The continuous directed-adjacency score matrix is then obtained through the polar product $\widehat{\bS}=\widehat{\bU}_{k}\widehat{\bP}_{k}\in \Real^{N \times N}$. The resulting $\widehat{\bS}$ assigns a continuous score to each ordered node pair and is converted into a binary adjacency matrix through the size-conditioned discretization. To obtain a binary adjacency matrix, PolarFlow estimates the size-conditioned edge density $\rho_{N}$ from the training set, i.e., $\rho_N = \E_{\bA\sim p_{\mathrm{data}}(\cdot\mid N)}[\frac{1}{N(N-1)}\sum_{i\neq j} A_{ij}]$. The corresponding edge budget is $m_N =\operatorname{round}\left(\rho_{N} N(N-1)\right)$. PolarFlow constructs $\widehat{\bA}$ by assigning ones to the $m_N$ largest off-diagonal entries of $\widehat{\bS}$ and zeros to all remaining entries. Self-loops are excluded by setting $\widehat A_{ii}=0$ for all $i$. The resulting $\widehat{\bA}\in\{0,1\}^{N\times N}$ is the binary adjacency matrix of the generated directed graph. \paragraph*{Categorical node labels} We model $p_\psi(\by\mid\widehat{\bA})$ using a separate masked categorical predictor, with training based on discrete flow matching \citep{qin2025defog} and sampling based on confidence-guided iterative unmasking \citep{chang2022maskgit}. The generated adjacency matrix $\widehat{\bA}$ remains fixed throughout label generation; details are provided in \cref{app.label_decoding}.

\section{Expanding Directed-Graph Dataset Diversity}
\label{sec:datasets}
\vspace{-2mm}
Existing benchmarks for directed-graph generation remain limited in structural diversity (e.g., sparsity) and domain coverage, restricting evaluation to a relatively narrow range of directed connectivity patterns. To broaden this scope, we construct two additional benchmark datasets, \emph{Circuits} and \emph{Twitter}, representing engineered functional systems and naturally evolved social systems, respectively. 
\vspace{-2.5mm}
\paragraph{Circuits: engineered functional networks.} We construct the Circuits dataset from two open-source hardware designs: the Ibex RISC-V CPU core and an Ethernet MAC controller. After compiling the designs into FPGA cell-level netlists of AMD Xilinx primitives, we represent each netlist as a directed cell graph and extract radius-two ego-subgraphs by traversing both incoming and outgoing connections while retaining feedback. Each node denotes a primitive cell instance and carries one of seven categorical cell-type labels, while an edge $u\rightarrow v$ indicates that the signal produced by cell $u$ drives an input of cell $v$. Parallel pin-level connections are collapsed into binary edges, and self-loops are removed. The dataset contains 318 graphs, split into 251 training, 41 validation, and 26 test graphs. Further construction details and dataset statistics are provided in \cref{app.Circuit_datasets}.
\vspace{-2.5mm}
\paragraph{Twitter: naturally evolved social networks.} We construct the Twitter dataset from the SNAP \emph{ego-Twitter} collection \citep{mcauley2012learning}, a directed social network collected from public follow lists. Each node represents an anonymized user account, and an edge $u\rightarrow v$ indicates that user $u$ follows user $v$, so in- and out-degree respectively characterize received attention and following activity. To prevent node leakage across splits, we partition the source accounts into disjoint training, validation, and test pools before extracting induced subgraphs through breadth-first expansion over both incoming and outgoing neighborhoods. This procedure yields 4{,}067 weakly connected graphs, divided into 3{,}447 training, 198 validation, and 422 test graphs. Further construction details and split-wise statistics are provided in \cref{app.twitter_datasets}.

Together, Circuits and Twitter broaden evaluation beyond a single directed-graph family by covering complementary connectivity regimes: label-dependent engineered connectivity and heterogeneous asymmetric social structure. Twitter contains no categorical node labels and is therefore evaluated only at the structural level; node-label generation is assessed on labeled datasets.
\vspace{-2mm}
\section{Experiments}
\label{sec:experiments}
\vspace{-2mm}
\subsection{Experimental setup}
\vspace{-2mm}
\paragraph{Datasets} We evaluate PolarFlow on the synthetic Stochastic Block Model (SBM) graphs and four real-world datasets: Visual Genome \citep{krishna2017visualgenome}, TPU Tiles \citep{google2023tpugraphs}, and our new Circuits and Twitter datasets. We follow the splits of \citep{carballocastro2026directo} for established benchmarks. Dataset construction, preprocessing and statistics are detailed in \cref{app.datasets}.
\vspace{-6mm}
\paragraph{Baselines}
We compare PolarFlow against an empirical maximum-likelihood (MLE) baseline over graph sizes, node types, and edge types; DiGress \citep{vignac2023digress}, a discrete diffusion model; DeFoG \citep{qin2025defog}, a discrete flow-matching model; and four DIRECTO variants \citep{carballocastro2026directo} combining DIRECTO or DIRECTO-DD with RRWP or MagLap positional encodings. For TPU Tiles, we also include the DAG-specific generators D-VAE \citep{zhang2019dvae} and LayerDAG \citep{li2025layerdag}.
\vspace{-3mm}
\paragraph{Setup}
We parameterize the joint velocity field $v_{\phi}$ using a residual MLP with a backbone shared across its four components. We randomly initialize the model parameters and train using AdamW optimizer. For each dataset, we report the mean and standard deviation over five independent sampling runs. Implementation and hyperparameter details are provided in \cref{app:implementation}.
\vspace{-2mm}
\paragraph{Evaluation metrics} Following standard protocols
\citep{martinkus2022spectre,thompson2022evaluation,bergmeister2024efficient}, we report the validity--uniqueness--novelty (V.U.N.) ratio and an average ratio based on maximum mean discrepancy (MMD). 
For labeled datasets, joint node--edge alignment is evaluated using Triplet MMD over dataset-specific semantic-tuple frequency distributions and radial basis function (RBF)-kernel MMD over GNN embeddings that jointly encode node labels and directed structure. Higher V.U.N.\ and lower MMD-based values indicate better performance. Full definitions and dataset-specific validity criteria are provided in \cref{app.eval.metric}.

\vspace{-1mm}
\subsection{Generative performance}
\vspace{-1mm}
\paragraph{Structural generation}
 \Cref{tab:structural_perform} reports the average MMD ratio for structural distributional alignment and V.U.N. for sample-level quality. PolarFlow is competitive with strong Directo variants on SBM and Twitter, and achieves the best average ratio on Circuits ($1.70$) with a perfect uniqueness--novelty score. Its largest gains occur on Visual Genome and TPU Tiles, reducing the best Directo ratios from $6.02$ and $18.46$ to $3.32$ and $11.54$, respectively. The V.U.N.\ gap on Visual Genome is consistent with its strict typed-edge validity criterion: because PolarFlow generates structures and labels in two stages without hard constraint enforcement, structurally well-aligned samples may still fail the graph-level test due to type incompatibilities. Overall, PolarFlow combines strong structural alignment with competitive
graph-level performance under the applicable validity, uniqueness, and
novelty criteria, while offering substantial training and sampling
efficiency advantages over discrete baselines (\cref{tab:training_inference_storage}). Complete descriptor-level MMD results are provided in \cref{app.complete_results}.
\vspace{-3mm}
\begin{table*}[!h]
\centering
\caption{Structural performance of directed graph generation methods (mean $\pm$ standard deviation over five sampling runs). Best results are bolded.
}
\label{tab:structural_perform}

\footnotesize
\setlength{\tabcolsep}{3pt}
\renewcommand{\arraystretch}{0.9}
\setlength{\aboverulesep}{0.20ex}
\setlength{\belowrulesep}{0.20ex}
\textbf{(a) Average structural MMD ratio ($\downarrow$)}
\par\vspace{2pt}

\begin{tabular*}{\textwidth}{
    @{\extracolsep{\fill}}lccccc@{}
}
\midrule
\textbf{Model}
& Circuit
& TPU Tiles
& Visual Genome
& SBM
& Twitter
\\
\midrule

\emph{Training set}
& $1.00$
& $1.00$
& $1.00$
& $1.00$
& $1.00$
\\
\midrule

MLE
& $6.49 \pm 0.65$
& $95.54 \pm 0.93$
& $17.00 \pm 0.60$
& $11.60 \pm 0.20$
& $12.59 \pm 0.18$
\\

D-VAE
& --
& OOT
& --
& --
& --
\\

LayerDAG
& --
& $135.14 \pm 10.23$
& --
& --
& --
\\

DiGress
& $2.10 \pm 0.13$
& $27.94 \pm 0.67$
& $17.00 \pm 0.60$
& $3.90 \pm 0.90$
& $1.34 \pm 0.02$
\\

DeFoG
& $1.90 \pm 0.24$
& $26.93 \pm 0.53$
& $10.60 \pm 0.80$
& $4.30 \pm 0.80$
& $1.48 \pm 0.04$
\\

\midrule

\textsc{Directo}-DD RRWP
& $1.95 \pm 0.44$
& $27.02 \pm 0.97$
& $15.30 \pm 0.80$
& $1.70 \pm 0.40$
& \bestcell{1.14 \pm 0.21}
\\

\textsc{Directo}-DD MagLap
& $2.06 \pm 0.38$
& $18.46 \pm 0.87$
& $7.60 \pm 0.70$
& \bestcell{1.50 \pm 0.40}
& $1.31 \pm 0.10$
\\

\textsc{Directo} RRWP
& $1.90 \pm 0.21$
& $25.96 \pm 0.42$
& $12.80 \pm 0.60$
& $1.80 \pm 0.50$
& $1.15 \pm 0.04$
\\

\textsc{Directo} MagLap
& $1.77 \pm 0.25$
& $20.96 \pm 1.35$
& $6.02 \pm 0.51$
& $2.00 \pm 0.30$
& $1.37 \pm 0.11$
\\

\midrule

\rowcolor{gray!10}
PolarFlow
& \bestcell{1.70 \pm 0.08}
& \bestcell{11.54 \pm 0.68}
& \bestcell{3.32 \pm 0.53}
& $2.00 \pm 0.30$
& $1.33 \pm 0.04$
\\

\bottomrule
\end{tabular*}

\vspace{2pt}

\textbf{(b) Validity, uniqueness, and novelty (V.U.N.)
($\%$, $\uparrow$)}
\par\vspace{2pt}

\begin{tabular*}{\textwidth}{
    @{\extracolsep{\fill}}lccccc@{}
}
\midrule
\textbf{Model}
& \shortstack{
    Circuit
  }
& \shortstack{
    TPU Tiles
  }
& \shortstack{
    Visual Genome
  }
& \shortstack{
    SBM
  }
& \shortstack{
    Twitter
  }
\\
\midrule

\emph{Training set}
& $0.00$
& $0.00$
& $0.00$
& $0.00$
& $0.00$
\\
\midrule

MLE
& \bestcell{100.00 \pm 0.00}
& $52.27 \pm 2.32$
& $90.81 \pm 0.92$
& $0.00 \pm 0.00$
& $100.00 \pm 0.00$
\\

D-VAE
& --
& OOT
& --
& --
& --
\\

LayerDAG
& --
& \bestcell{99.93 \pm 1.22}
& --
& --
& --
\\

DiGress
& $98.46 \pm 1.88$
& $86.63 \pm 1.06$
& $96.59 \pm 0.55$
& $41.50 \pm 5.10$
& $100.00 \pm 0.00$
\\

DeFoG
& $99.23 \pm 1.54$
& $79.20 \pm 0.85$
& $97.81 \pm 0.62$
& $37.00 \pm 6.60$
& $100.00 \pm 0.00$
\\

\midrule

\textsc{Directo}-DD RRWP
& $96.17 \pm 3.21$
& $88.09 \pm 2.28$
& $99.75 \pm 0.81$
& $81.50 \pm 3.20$
& $100.00 \pm 0.00$
\\

\textsc{Directo}-DD MagLap
& \bestcell{100.00 \pm 0.00}
& $81.74 \pm 1.69$
& $99.84 \pm 0.06$
& $95.50 \pm 3.70$
& $100.00 \pm 0.00$
\\

\textsc{Directo} RRWP
& $99.23 \pm 1.53$
& $78.34 \pm 1.17$
& \bestcell{99.94 \pm 0.05}
& \bestcell{99.50 \pm 1.00}
& $100.00 \pm 0.00$
\\

\textsc{Directo} MagLap
& $96.17 \pm 3.21$
& $82.69 \pm 0.98$
& $99.08 \pm 0.71$
& $96.50 \pm 2.50$
& $100.00 \pm 0.00$
\\

\midrule

\rowcolor{gray!10}
PolarFlow
& \bestcell{100.00 \pm 0.00}
& $82.45 \pm 1.22$
& $91.65 \pm 1.04$
& $92.50 \pm 3.80$
& $100.00 \pm 0.00$
\\

\bottomrule
\end{tabular*}

\vspace{3pt}

\begin{minipage}{\textwidth}
\footnotesize
\raggedright
\textit{Note.}
In (b), Circuit and Twitter report the percentage of samples that are
both unique and novel, as no binary validity criterion is defined for
these datasets. D-VAE and LayerDAG are DAG-specific and are evaluated only on
TPU Tiles. OOT indicates that a model could not be run within a reasonable
timeframe.
\end{minipage}
\end{table*}
\vspace{-2mm}
\paragraph{Node-label generation} \Cref{tab:node-label_performance}
reports joint node--edge distributional alignment on the three labeled real-world datasets. PolarFlow achieves the lowest RBF MMD on Visual Genome and remains comparable with DiGress, DeFoG, and the Directo variants on Circuits
and TPU Tiles, outperforming MLE on both datasets and LayerDAG on the
latter. Structurally faithful sampled adjacencies provide informative conditioning for label generation, potentially reducing train--inference mismatch and benefiting RBF MMD, whose embeddings jointly encode topology and labels. However, the remaining Triplet MMD gaps indicate that structural fidelity alone does not capture label-conditioned connectivity. Additional label-alignment results are available in \cref{app.complete_results}.

\begin{table*}[t]
\centering
\caption{
Joint node--edge distributional alignment on three labeled real-world
datasets.
}
\label{tab:node-label_performance}

\fontsize{7.5pt}{9pt}\selectfont
\setlength{\tabcolsep}{1pt}
\renewcommand{\arraystretch}{0.9}
\setlength{\aboverulesep}{0.19ex}
\setlength{\belowrulesep}{0.19ex}

\begin{tabular}{@{}l*{6}{c}@{}}
\midrule
&
\multicolumn{2}{c}{\textbf{Circuit}}
&
\multicolumn{2}{c}{\textbf{TPU Tiles}}
&
\multicolumn{2}{c}{\textbf{Visual Genome}}
\\
\cmidrule(lr){2-3}
\cmidrule(lr){4-5}
\cmidrule(lr){6-7}

\textbf{Model}
& \textbf{Triplet MMD}$\downarrow$
& \textbf{RBF MMD}$\downarrow$
& \textbf{Triplet MMD}$\downarrow$
& \textbf{RBF MMD}$\downarrow$
& \textbf{Triplet MMD}$\downarrow$
& \textbf{RBF MMD}$\downarrow$
\\
\midrule

\emph{Training set}
& $0.0117$
& $0.0918$
& $0.00061$
& $0.0056$
& $0.0123$
& $0.0257$
\\
\midrule

MLE
& $0.035 \pm 0.002$
& $0.160 \pm 0.033$
& $0.091 \pm 0.004$
& $0.158 \pm 0.013$
& $0.330 \pm 0.004$
& $0.121 \pm 0.009$
\\

D-VAE
& --
& --
& OOT
& OOT
& --
& --
\\

LayerDAG
& --
& --
& $0.474 \pm 0.011$
& $0.422 \pm 0.077$
& --
& --
\\

DiGress
& $0.018 \pm 0.012$
& $0.097 \pm 0.011$
& $0.005 \pm 0.001$
& $0.016 \pm 0.001$
& $0.230 \pm 0.010$
& $0.083 \pm 0.012$
\\

DeFoG
& $0.021 \pm 0.007$
& $0.098 \pm 0.005$
& $0.004 \pm 0.001$
& $0.017 \pm 0.001$
& $\mathbf{0.016 \pm 0.004}$
& $0.074 \pm 0.002$
\\

\midrule

\textsc{Directo}-DD RRWP
& $\mathbf{0.013 \pm 0.005}$
& $\mathbf{0.087 \pm 0.018}$
& $\mathbf{0.002 \pm 0.000}$
& $\mathbf{0.012 \pm 0.001}$
& $0.027 \pm 0.009$
& $0.053 \pm 0.011$
\\

\textsc{Directo}-DD MagLap
& $0.015 \pm 0.006$
& $0.103 \pm 0.018$
& $0.004 \pm 0.001$
& $0.015 \pm 0.000$
& $0.032 \pm 0.005$
& $0.063 \pm 0.011$
\\

\textsc{Directo} RRWP
& $0.014 \pm 0.006$
& $0.090 \pm 0.01$
& $0.002 \pm 0.000$
& $0.013 \pm 0.001$
& $0.048 \pm 0.016$
& $0.065\pm0.052$
\\

\textsc{Directo} MagLap
& $0.016 \pm 0.005$
& $0.090 \pm 0.010$
& $0.006 \pm 0.001$
& $0.016 \pm 0.001$
& $0.035 \pm 0.005$
& $0.080 \pm 0.007$
\\

\midrule

\rowcolor{gray!12}
PolarFlow
& $0.023 \pm 0.009$
& $0.108 \pm 0.014$
& $0.009 \pm 0.001$
& $0.024 \pm 0.002$
& $0.070 \pm 0.007$
& $\mathbf{0.046 \pm 0.007}$
\\

\bottomrule
\end{tabular}
\vspace{-4mm}
\end{table*}

\vspace{-2mm}
\paragraph{Ablation studies (cf.\ \cref{app.ablation})}\label{sec:ablations}
We examine two design choices supporting PolarFlow's stability and coherence: (i) applying a small diagonal perturbation to the data endpoint before polar--spectral decomposition to mitigate rank and spectral degeneracies, and (ii) jointly constructing the four source channels from a shared noisy matrix to preserve their cross-channel coupling. The diagonal perturbation generally improves structural distributional alignment, reported in \cref{tab:diagonal_shift_ablation}, providing empirical evidence consistent with the stability analysis in \cref{thm:stability} and indicating that appropriately chosen endpoint regularization can improve generation quality. For source construction, deriving the four channels through the polar--spectral decomposition of a shared source matrix yields lower average MMD ratios than a factorized prior, in which the channels are sampled independently and no common matrix is decomposed at $t=0$. Moreover, a coupled variant based on a Gaussian-only source matrix yields results close to those of the structured-template variant, suggesting that the improvement stems primarily from cross-channel coupling rather than template-specific structure, shown in \cref{tab:source_construction_ablation}. Complete formulations and results are provided in \cref{app.ablation}.

\vspace{-1mm}
\subsection{Computational efficiency}
\label{sec:complexity}
\vspace{-1mm}


PolarFlow and discrete flow-matching methods differ in how they represent graph structure and perform sampling. Discrete methods predict dense categorical distributions over all ordered node pairs and, during generation, repeatedly evaluate transition probabilities or rates and update categorical edge states. PolarFlow instead evolves four continuous polar--spectral channels through a joint velocity field, reconstructing and discretizing the adjacency matrix only at the endpoint. To assess the practical consequences of these computational profiles, \cref{tab:training_inference_storage} reports per-epoch
training time and per-graph sampling time. PolarFlow achieves substantially shorter training and sampling times than the evaluated discrete generative methods while maintaining comparable generation quality, an efficiency advantage for sample-intensive tasks such as large-scale evaluation and synthetic-data augmentation.
Additional runtime and resource comparisons are provided in \cref{app.runtime}.

\vspace{-3mm}
\begin{table*}[!htb]
\centering
\setlength{\tabcolsep}{2pt}
\renewcommand{\arraystretch}{0.85}
\caption{Comparison of training time (s/epoch), sampling time (s/graph), and storage (GB).}
\resizebox{\textwidth}{!}{
\begin{tabular}{l|ccc|ccc|ccc}
\midrule
\multirow{2}{*}{Method} & \multicolumn{3}{c|}{Training Time (s)} & \multicolumn{3}{c|}{Sampling Time (s)} & \multicolumn{3}{c}{Storage (GB)} \\
                        & Twitter   & Circuit  & TPU Tiles  & Twitter   & Circuit  & TPU Tiles   & Twitter   & Circuit  & TPU Tiles \\
\midrule
Digress                  & 21.56      &  1.47    & 16.53          & 1.65  &  1.38    & 1.51            & 8.89    & 7.89     & 9.42 \\
DeFoG                    & 16.85      &  1.13    & 13.27          & 1.80   & 1.91    & 1.79            & 8.96    & 8.15     & 9.13 \\
\midrule
\textsc{Directo}-DD RRWP & 29.67      &  1.45    & 22.86          & 1.69   & 1.00    & 1.67            & 18.67    & 17.79   & 19.67 \\
\textsc{Directo}-DD MagLap & 45.19    &  2.09    & 32.39          & 3.77   & 2.17    & 3.41            & 18.63    & 17.79   & 20.07 \\
\textsc{Directo} RRWP      & 29.67    &  1.42    & 22.79         & 3.56   & 1.88    & 3.23            & 18.66    & 18.27   & 19.55 \\
\textsc{Directo} MagLap    & 44.63    &  2.09    & 32.52        & 7.78   & 6.66    & 6.91            & 18.63    & 18.27   & 19.93 \\
\midrule
Ours                       & 3.67    &  1.10   & 3.33         & 0.19   & 0.16     & 0.14           & 1.29     & 1.03    & 1.13  \\
\bottomrule
\end{tabular}
}
\label{tab:training_inference_storage}
\vspace{-4mm}
\end{table*}

\section{Conclusion}
\label{sec:conclusion}
\vspace{-1mm}
We introduce \emph{PolarFlow}, a continuous geometric framework for directed graph generation based on polar decomposition and truncated spectral coordinates. By jointly modeling the orthogonal and positive-semidefinite factors through four coupled channels, PolarFlow replaces edge-wise categorical generation with geometry-aware flow matching over a structured representation of the adjacency matrix. Experiments across synthetic and real-world datasets, including the newly introduced Circuit and Twitter benchmarks, demonstrate that PolarFlow achieves generative quality comparable to competitive discrete methods while requiring substantially less training and sampling time in practice. These results establish polar--spectral modeling as a promising alternative to discrete edge-wise generation and motivate its broader use in application-specific directed graph design.

\section*{Acknowledgments}

ACC was supported by the Swiss National Science Foundation (SNSF) under grant~10001445 and in part by the EPFL AI Center and Swiss AI Initiative PhD Fellowship.


\clearpage

\appendix
\section{Related work}\label[Appendix]{app.related_works}

\paragraph{Graph generative models.}
Autoregressive models build a graph through a sequence of node and edge decisions \citep{you2018graphrnn,liao2019efficient,dai2020BiGG}, while one-shot alternatives include variational autoencoders \citep{kipf2016variational,simonovsky2018graphvae,jin2018JT-VAE,ma2021gfvae,vignac2022topn}, generative adversarial networks \citep{decao2018molgan,krawczuk2021gggan}, and normalizing flows \citep{madhawa2019graphnvp,wehenkel2021graphical,lippe2021categorical,luo2021graphdf}. Diffusion methods first adapted continuous dynamics to graphs \citep{niu2020permutation,luo2021predicting,xu2022geodiff,jo2022GDSS} and then moved to categorical state spaces and their continuous-time and flow-matching variants \citep{austin2021structured,haefeli2022diffusion,vignac2023digress,jo2024Grum,xu2024disco,siraudin2024cometh,campbell2024generative,gat2024discrete,qin2025defog}. Across these families, the generative variables are node categories and edge categories over all candidate pairs.

\paragraph{Spectral and geometric generation.}
Latent graph generators reduce the dimension of the visible state through a learned representation \citep{zhou2024unifying,fu2024hyperbolic,fan2024graphusion,siraudin2026principledlatentdiffusiongraphs}. Spectral methods instead derive coordinates from graph operators.  SPECTRE generates leading Laplacian eigenpairs \citep{martinkus2022spectre}, and GSDM uses low-rank stochastic dynamics on graph spectra \citep{luo2023GSDM}; both
focus on undirected graphs with symmetric operators.  Riemannian generative models show how conditional paths can respect non-Euclidean domains
\citep{chen2024Geomflow,lou2023scaling}.  PolarFlow combines these ideas for an asymmetric adjacency matrix by separating its polar orientation and magnitude.

\paragraph{Directed graph generation.}
DAG generators usually exploit a topological ordering, including D-VAE \citep{zhang2019dvae}, GraphPNAS \citep{li2022graphpnas}, LayerDAG
\citep{li2025layerdag}, and SeaDAG \citep{zhou2024seadag}.  
For general digraphs, DGDK uses a heat-kernel denoising autoencoder \citep{law2025DGDK}, while Directo combines discrete flow matching with asymmetric positional encodings and dual attention \citep{carballocastro2026directo}, and LG-Flow~\citep{siraudin2026principledlatentdiffusiongraphs} uses latent diffusion.  PolarFlow differs in the generative variables themselves: it transports a global polar--spectral state rather than categorical variables indexed by ordered pairs.

\section{Discussions and proofs} \label[Appendix]{sec:pf}
To start with, we introduce \emph{regularity assumptions} on $\bA_\varepsilon=\bA+\bE_\varepsilon=\bU\bP$.

For $1\leq k<N$, let $\mathcal I_p$ and $\mathcal I_u$ index the retained
eigenvalues of $\bP$ and $\bU$, respectively. Define the spectral separations between the retained and discarded spectra of $\bP$ and $\bU$, respectively, by
$\gamma_{p,\varepsilon}:=\min_{\substack{i\in\mathcal I_p, j\notin\mathcal I_p}}|\lambda_i-\lambda_j|$, and $\gamma_{u,\varepsilon}
    :=
    \min_{\substack{i\in\mathcal I_u, j\notin\mathcal I_u}}
    \left|e^{\mathrm{i}\theta_i}-e^{\mathrm{i}\theta_j}\right|.$ We say that $\bA_{\epsilon}$ is \emph{spectrally regular} if the following holds:
\begin{enumerate}
    \item [(i)] The rank margin is positive, $ \alpha_\varepsilon:=\sigma_{\min}(\bA_\varepsilon)>0$.
    \item [(ii)] The relevant spectral gaps are positive, $\gamma_{p,\varepsilon}>0$ and $\gamma_{u,\varepsilon}>0$;
    \item [(iii)] The score-selection margin is positive, $\mu_\varepsilon:=s_{k}-s_{k+1}>0$.
\end{enumerate}

For example, if both $\bU$ and $\bP$ do not have repeated eigenvalues, then $\bA_{\epsilon}$ is spectrally regular. In practice, the matrix $\bA$ can be irregular, particularly when the graph is sparse. Intuitively, the purpose of continuously perturbing $\bA$ matrix with $\bE_{\epsilon}$ is to make $\bA_{\epsilon}$ spectrally regular (with high probability).

Let \(\bZ_0\in\mathcal M_k\) be a fixed source endpoint and let \(\bZ_1=\Phi_k(\bA_\epsilon)\in\mathcal M_k\) be a target endpoint. We say that \(\bA_\epsilon\) is \emph{geometrically regular} (w.r.t.\ \(\bZ_0\)) if the following holds:
\begin{enumerate}[(i)]
\item There exists a neighborhood \(\mathcal U\subset\mathcal M_k\) of \(\bZ_1\) such that for \(c\in\{p,u\}\), 
the Riemannian logarithm
   \[
   (\bQ_{c,0},\bQ_c)\mapsto \operatorname{Log}_{\bQ_{c,0}}(\bQ_c)
   \]
   is uniquely defined and smooth on this neighborhood. 

\item For every retained angle, the endpoint \(\theta_{1}^{(k)}\) satisfies \(\theta_{1}^{(k)}-\theta_{0}^{(k)} \notin (2\mathbb{Z}+1)\pi\).
\end{enumerate}

Intuitively, spectral regularity ensures the stability of the matrix decomposition, while geometric regularity ensures the local uniqueness and smoothness of the geodesic interpolation on the underlying manifolds. 

$\bA_{\epsilon}$ is \emph{regular} if it is both spectrally and geometrically regular. 

\begin{proof}[Proof of \cref{thm:stability}]
Let $ \bA_\varepsilon=\bU\bP$ and $\widetilde{\bA}_\varepsilon =\widetilde{\bU}\widetilde{\bP}$ be the polar decompositions of the two regularized matrices, and set $\delta:=\|\widetilde{\bA}_\varepsilon-\bA_\varepsilon\|_2$ and $\alpha_\varepsilon=\sigma_{\min}(\bA_\varepsilon)>0$. We prove the result by first controlling the perturbations of the polar factors and then propagating these bounds through the retained spectral coordinates, the conditional paths, and finally the learned velocity field.

\emph{Step 1 Polar-factor stability:} We first show that the polar decomposition remains unique in a sufficiently
small neighborhood of $\bA_\varepsilon$, and then control the variation of its two polar factors. Weyl’s singular-value inequality \citep{horn2012matrix} gives $ \sigma_{\min}(\widetilde{\bA}_\varepsilon)\geq \alpha_\varepsilon-\delta$. Hence, by choosing $r<\alpha_\varepsilon/2$, every perturbation with
$\delta<r$ satisfies $\sigma_{\min}(\widetilde{\bA}_\varepsilon)\geq\frac{\alpha_\varepsilon}{2}>0$. Consequently, both $\bA_\varepsilon$ and $\widetilde{\bA}_\varepsilon$ are nonsingular. Their positive polar factors are therefore positive definite, and their orthogonal polar factors are uniquely determined.

Since the matrices are $N\times N$, the spectral and Frobenius norms satisfy $\|\bX\|_F\leq\sqrt{N}\|\bX\|_2$. The Frobenius-norm perturbation bound of \citep{chen1989perturbation} for positive polar factors thus gives
\begin{equation}
    \left\|\widetilde{\bP}-\bP\right\|_F
    \leq
    \sqrt{2} \left\|
        \widetilde{\bA}_\varepsilon-\bA_\varepsilon
    \right\|_F
    \leq
    \sqrt{2N}\,\delta.
    \label{eq:proof_positive_factor}
\end{equation}
For the orthogonal polar factors, Theorem 1 of \citet{li1995new} yields
\begin{align}
    \left\|
        \widetilde{\bU}-\bU
    \right\|_F
    &\leq
    \frac{
        2
        \left\|
            \widetilde{\bA}_\varepsilon-\bA_\varepsilon
        \right\|_F
    }{
        \sigma_{\min}(\bA_\varepsilon)
        +
        \sigma_{\min}(\widetilde{\bA}_\varepsilon)
    }
    \leq
    \frac{4\sqrt{N}}{3\alpha_\varepsilon}\,\delta.
    \label{eq:proof_orthogonal_factor}
\end{align}
\Cref{eq:proof_positive_factor} and \cref{eq:proof_orthogonal_factor} show that both polar factors vary Lipschitz-continuously with the regularized matrix. We next transfer the polar-factor bounds to the four PolarFlow endpoint channels.     

\emph{Step 2 Stability of the retained spectral representation.}
Having controlled the polar factors, we now transfer these bounds to the retained eigenvalues, eigenangles, and spectral frames. Using
$\|\bX\|_2\leq\|\bX\|_F$ and the estimates from Step 1, choose $r>0$ sufficiently small so that every perturbation satisfying $\delta<r$ obeys
\begin{equation}
\label{eq:proof_u_gap_preservation}
\begin{aligned}
\left\|\widetilde{\bP}-\bP\right\|_2
    &\leq
    \left\|\widetilde{\bP}-\bP\right\|_F
    \leq
    \sqrt{2N}\,\delta
    <
    \frac{\gamma_{p,\varepsilon}}{2},\\
    \left\|\widetilde{\bU}-\bU\right\|_2
    &\leq
    \left\|\widetilde{\bU}-\bU\right\|_F
    \leq
    \frac{4\sqrt{N}}{3\alpha_\varepsilon}\,\delta
    <
    \frac{\gamma_{u,\varepsilon}}{2}.
\end{aligned}
\end{equation}
These inequalities ensure that the retained and discarded parts of each spectrum remain separated under the perturbation. Consequently, the retained invariant subspaces of $\bP$ and $\bU$ can be matched consistently with their perturbed counterparts, and their dimensions remain unchanged.

For the positive polar factor, recall that $\bP = (\bA_{\varepsilon}^{*}\bA_{\varepsilon})^{1/2}$. Its eigenvalues are therefore the singular values of $\bA_{\varepsilon}$. Under the common ordering of the retained eigenvalues, Mirsky’s singular-value perturbation theorem
\citep{mirsky1960symmetric} gives
\begin{align*}
    \|\widetilde{\blambda}_{p}^{(k)}-\blambda_{p}^{(k)}\|_{2}\leq \|\widetilde{\bA}_{\varepsilon}-\bA_{\varepsilon}\|_{F} \leq \sqrt{N}\delta
\end{align*}
Thus, the retained eigenvalue channel is Lipschitz-continuous with respect to the matrix perturbation. The same spectral separation also stabilizes the associated retained invariant subspace. Indeed, the separation between the perturbed retained spectrum and the complementary spectrum of $\bP$ is bounded below by 
\begin{align*}
    \gamma_{p,\varepsilon} - \|\widetilde{\bP}-\bP\|_{2} > \frac{\gamma_{p,\varepsilon}}{2}.
\end{align*}
To remove the intrinsic basis ambiguity, we align the perturbed frame
separately within each retained eigenspace. Let
$\nu_1,\ldots,\nu_m$ denote the distinct retained eigenvalues of $\bP$,
with multiplicities $m_1,\ldots,m_m$, and let $\bQ_{p,\ell}$ and
$\widetilde{\bQ}_{p,\ell}$ be orthonormal bases for the corresponding
unperturbed and perturbed spectral subspaces. For each $\ell$, choose
\begin{equation}
    \bR_{p,\ell}^{\star}
    \in
    \argmin_{\bR\in\mathrm{O}(m_\ell)}
    \left\|
        \widetilde{\bQ}_{p,\ell}\bR-\bQ_{p,\ell}
    \right\|_F .
    \label{eq:p_frame_alignment}
\end{equation}
For a simple eigenvalue, this alignment only fixes the eigenvector sign;
for a repeated eigenvalue, it selects the closest orthonormal basis
within the corresponding spectral subspace. The positive spectral gaps
ensure that the perturbed subspaces have the same dimensions as their
unperturbed counterparts. Moreover, for sufficiently small $\delta$,
their principal angles are smaller than $\pi/2$, so the cross-Gram
matrices $\widetilde{\bQ}_{p,\ell}^{*}\bQ_{p,\ell}$ are nonsingular
and the minimizing orthogonal transformations in
\cref{eq:p_frame_alignment} are unique and vary continuously with the
perturbation. We henceforth replace each
$\widetilde{\bQ}_{p,\ell}$ by
$\widetilde{\bQ}_{p,\ell}\bR_{p,\ell}^{\star}$ and retain the notation
$\widetilde{\bQ}_{p}$ for the resulting aligned frame.

With this convention, applying the Davis–Kahan theorem \citep{davis1970rotation} and the standard relation between subspace distance and optimally aligned frame distance gives
\begin{align*}
    \|\widetilde{\bQ}_{p}-\bQ_{p}\|_{F} \leq \frac{2\sqrt{2}}{\gamma_{p,\varepsilon}}\|\widetilde{\bP}-\bP\|_{F} \leq \frac{4\sqrt{N}}{\gamma_{p,\varepsilon}}\delta
\end{align*}
Therefore, both retained channels associated with $\bP$ are locally Lipschitz-stable.

We now consider the orthogonal polar factor. Let $z_{i} = e^{\mathrm{i}\theta_i}$ and $\widetilde{z}_{i} = e^{\mathrm{i}\widetilde{\theta}_i}$ denote the eigenvalues of $\bU$ and $\widetilde{\bU}$, respectively. Since both matrices are normal, the Hoffman–Wielandt theorem \citep{hoffman1953variation} permits the perturbed eigenvalues such that 
\begin{align*}
    \left(\sum_{i=1}^{N}|\widetilde{z}_{i}-z_{i}|^{2}\right)^{1/2} \leq \|\widetilde{\bU}-\bU\|_{F}.
\end{align*}
Moreover, because $\|\widetilde{\bU}-\bU\|_{F} < \gamma_{u,\varepsilon}/2$, the retained and discarded spectra remain separated. Restricting the above inequality to the retained eigenvalues therefore gives 
\begin{align*}
    \left(\sum_{i\in \calI_{u}}|\widetilde{z}_{i}-z_{i}|^{2}\right)^{1/2} \leq \|\widetilde{\bU}-\bU\|_{F}\leq     \frac{4\sqrt{N}}{3\alpha_\varepsilon}\,\delta.  
\end{align*}
Let $\wrap(x)\in [-\pi,\pi]$ denote the principal value of $x$ module $2\pi$. For each $i\in \calI_{u}$, we have
\begin{align*}
  |\widetilde{z}_{i}-z_{i}| = 2 \bigg|\sin\left(\frac{\wrap(\widetilde{\theta}_{i}-\theta_{i})}{2}\right)\bigg|.
\end{align*}
since $2\left|\sin\left(\frac{x}{2}\right)\right| \geq \frac{2|x|}{\pi}$ for $ x\in[-\pi,\pi]$, we have
\begin{align*}
    \left|
        \operatorname{wrap}
        \left(\widetilde{\theta}_i-\theta_i\right)
    \right|
    \leq
    \frac{\pi}{2}
    \left|\widetilde{z}_i-z_i\right|
\end{align*}
for each $i\in\calI_u$. Applying this inequality over the retained indices gives
\begin{align}
    \left\|
        \operatorname{wrap}
        \left(
            \widetilde{\boldsymbol{\theta}}_k
            -
            \boldsymbol{\theta}_k
        \right)
    \right\|_2
    &\leq
    \frac{\pi}{2}
    \left(
        \sum_{i\in\calI_u}
        \left|\widetilde{z}_i-z_i\right|^2
    \right)^{1/2}
   \leq
    \frac{2\pi\sqrt{N}}
         {3\alpha_\varepsilon}\,\delta.
    \label{eq:proof_eigenangle_stability}
\end{align}

The same spectral separation also stabilizes the associated retained
frame. In particular, the perturbed retained spectrum remains separated
from the complementary spectrum of $\bU$ by at least
\begin{equation}
    \gamma_{u,\varepsilon}
    -
    \left\|
        \widetilde{\bU}-\bU
    \right\|_2
    >
    \frac{\gamma_{u,\varepsilon}}{2}.
\end{equation}
Applying the invariant-subspace perturbation bound for normal matrices
and using the consistent frame alignment assumed in the theorem yields
\begin{align}
    \left\|
        \widetilde{\bQ}_{u}
        -
        \bQ_{u}
    \right\|_F
    &\leq
    \frac{2\sqrt{2}}
         {\gamma_{u,\varepsilon}}
    \left\|
        \widetilde{\bU}-\bU
    \right\|_F
    \leq
    \frac{8\sqrt{2N}}
         {3\alpha_\varepsilon\gamma_{u,\varepsilon}}
    \,\delta.
    \label{eq:proof_u_frame_stability}
\end{align}
Finally, we establish the stability of the score-based mode selection. Under the locally consistent frame convention assumed above, each score $s_i=2\left|\sin\!\left(\frac{\theta_i}{2}\right)\right|\bq_{u,i}^{\ast}\bP\bq_{u,i}$ is locally Lipschitz in its spectral arguments. Indeed,
$2|\sin(\theta/2)|$ is $1$-Lipschitz with respect to the wrapped angular
distance, and $\|\mathbf q_i\|_2=1$. Consequently,
\begin{align}
    \left|
        \widetilde{s}_i-s_i
    \right|
    \leq&
    \left\|
        \widetilde{\bP}
    \right\|_2
    \left|
        \wrap
        \left(
            \widetilde{\theta}_i-\theta_i
        \right)
    \right|
    \nonumber\\
    &+
    2\left\|
        \widetilde{\bP}-\bP
    \right\|_2
    +
    4\left\|
        \bP
    \right\|_2
    \left\|
        \widetilde{\mathbf q}_i-\mathbf q_i
    \right\|_2.
    \label{eq:proof_score_perturbation}
\end{align}
The spectral estimates established above therefore imply that, within a sufficiently small neighborhood of $\bA_\varepsilon$, there exists a finite constant $L_s$, independent of $\delta$, such that
\begin{equation}
    \max_i
    \left|
        \widetilde{s}_i-s_i
    \right|
    \leq
    L_s\delta.
\end{equation}
Recall that the positive selection margin is $\mu_\varepsilon= s_{k}-s_{k+1}>0$ where the scores are ordered in nonincreasing order. Choose the neighborhood radius $r$ to additionally satisfy $    r
    < \frac{\mu_\varepsilon}{2L_s}$. Then, whenever $\delta<r$, every retained index $i$ and discarded
index $j$ satisfy
\begin{align*}
     \widetilde{s}_i-\widetilde{s}_j
    \geq
    s_i-s_j-2L_s\delta
    \geq
    \mu_\varepsilon-2L_s\delta
    >
    0.
\end{align*}
Thus, the perturbation does not change the top-$k$ index set.

With the top-$k$ index set preserved and the spectral frames consistently aligned, the preceding channel-wise bounds yield
\begin{equation}
    d_\Phi\left(
        \Phi_k(\bA_\varepsilon),
        \Phi_k(\widetilde{\bA}_\varepsilon)
    \right)
    \leq
    C
    \left\|
        \widetilde{\bA}_\varepsilon-\bA_\varepsilon
    \right\|_2.
    \label{eq:polarflow_endpoint_stability}
\end{equation}
Here, $C<\infty$ depends on $N$, $\alpha_\varepsilon$,
$\gamma_{p,\varepsilon}$, and $\gamma_{u,\varepsilon}$, but not on the perturbation magnitude. This proves the claimed local Lipschitz stability of the regularized PolarFlow endpoint.

\emph{Step 3: Stability of the conditional velocity targets.} Let $C$ denote the endpoint-stability constant established above. Geometric regularity implies that geodesic interpolation on the Stiefel manifolds or on the high-dimensional torus avoids the cut locus of the $\Log$ function or the discontinuity of the wrap function.  Therefore, in a sufficiently small neighborhood of \(\bZ_1\), there exist finite constants $L_Z$ and $L_{\dot Z}$ such that
\begin{align}
    \sup_{t\in[0,1]}
    \left\|
        \bZ_t(\widetilde{\bA}_\varepsilon)
        -
        \bZ_t(\bA_\varepsilon)
    \right\|
    &\leq
    L_Z\,
    d_\Phi\left(
        \Phi_k(\bA_\varepsilon),
        \Phi_k(\widetilde{\bA}_\varepsilon)
    \right),
    \label{eq:conditional_path_stability}\\
    \sup_{t\in[0,1]}
    \left\|
        \dot{\bZ}_t(\widetilde{\bA}_\varepsilon)
        -
        \dot{\bZ}_t(\bA_\varepsilon)
    \right\|
    &\leq
    L_{\dot Z}\,
    d_\Phi\left(
        \Phi_k(\bA_\varepsilon),
        \Phi_k(\widetilde{\bA}_\varepsilon)
    \right).
    \label{eq:conditional_path_velocity_lipschitz}
\end{align}
Choose $r$ sufficiently small that both conditional paths remain in
the neighborhoods on which these Lipschitz constants apply. Combining
\cref{eq:conditional_path_velocity_lipschitz} with the endpoint bound gives
\begin{align}
    \sup_{t\in[0,1]}
    \left\|
        \dot{\bZ}_t(\widetilde{\bA}_\varepsilon)
        -
        \dot{\bZ}_t(\bA_\varepsilon)
    \right\|
    &\leq
    L_{\dot Z} C
    \left\|
        \widetilde{\bA}_\varepsilon-\bA_\varepsilon
    \right\|_2.
    \label{eq:proof_conditional_velocity_stability}
\end{align}
Thus, the conditional velocity targets vary Lipschitz-continuously with the regularized input matrix.

\emph{Step 4: Stability of the learned velocity field.}
We now transfer the stability of the conditional velocity targets to the learned velocity field. For any regularized matrix $\bA'_\varepsilon$ in the neighborhood considered above, recall the velocity approximation error along its conditional path can be expressed as
\begin{equation}
    \boldsymbol e_\phi(\bA'_\varepsilon,t)
    :=
    \widehat{\bv}_\phi
    \left(
        \bZ_t(\bA'_\varepsilon),t
    \right)
    -
    \dot{\bZ}_t(\bA'_\varepsilon).
    \label{eq:velocity_approximation_error}
\end{equation}
Assume that this approximation error varies locally
Lipschitz-continuously with the regularized matrix, uniformly over $t\in[0,1]$. That is, there exists a finite constant $L_{e}$ such that 
\begin{equation}
    \sup_{t\in[0,1]}
    \left\|
        \boldsymbol e_\phi(\widetilde{\bA}_\varepsilon,t)
        -
        \boldsymbol e_\phi(\bA_\varepsilon,t)
    \right\|
    \leq
    L_e
    \left\|
        \widetilde{\bA}_\varepsilon-\bA_\varepsilon
    \right\|_2.
    \label{eq:velocity_error_stability}
\end{equation}
By \cref{eq:velocity_approximation_error}, for every $t\in[0,1]$,
\begin{align}
&\widehat{\bv}_\phi
\left(
    \bZ_t(\widetilde{\bA}_\varepsilon),t
\right)
-
\widehat{\bv}_\phi
\left(
    \bZ_t(\bA_\varepsilon),t
\right)
\nonumber\\
&\quad=
\dot{\bZ}_t(\widetilde{\bA}_\varepsilon)
-
\dot{\bZ}_t(\bA_\varepsilon)
+
\boldsymbol e_\phi(\widetilde{\bA}_\varepsilon,t)
-
\boldsymbol e_\phi(\bA_\varepsilon,t).
\end{align}
Taking norms and the supremum over $t$, and then applying
\cref{eq:proof_conditional_velocity_stability,eq:velocity_error_stability},
gives
\begin{align}
&\sup_{t\in[0,1]}
\left\|
    \widehat{\bv}_\phi
    \left(
        \bZ_t(\widetilde{\bA}_\varepsilon),t
    \right)
    -
    \widehat{\bv}_\phi
    \left(
        \bZ_t(\bA_\varepsilon),t
    \right)
\right\|
\nonumber\\
&\quad\leq
\sup_{t\in[0,1]}
\left\|
    \dot{\bZ}_t(\widetilde{\bA}_\varepsilon)
    -
    \dot{\bZ}_t(\bA_\varepsilon)
\right\|
+
\sup_{t\in[0,1]}
\left\|
    \boldsymbol e_\phi(\widetilde{\bA}_\varepsilon,t)
    -
    \boldsymbol e_\phi(\bA_\varepsilon,t)
\right\|
\nonumber\\
&\quad\leq
\left(
    L_{\dot Z}C_\Phi+L_e
\right)
\left\|
    \widetilde{\bA}_\varepsilon-\bA_\varepsilon
\right\|_2.
\label{eq:proof_learned_velocity_stability}
\end{align}
Therefore, taking $C:=L_{\dot Z}C_\Phi+L_e$
establishes \cref{eq:learned_velocity_stability} and completes the proof.
\end{proof}

\section{Dataset descriptions}
\label[Appendix]{app.datasets}
\subsection{Stochastic Block Model~\citep{holland1983sbm}}

We adopt the SBM sampling protocol of \citet{carballocastro2026directo} without modification, so that our numbers remain directly comparable to the results reported there. Every synthetic graph in our experiments is represented as a binary adjacency matrix $\bm{A}\in\{0,1\}^{n\times n}$, with $\bm{A}_{ij}=1$ indicating a directed edge from node $i$ to node $j$.  Sampled matrices are then converted into whichever input format the corresponding model expects. 

A graph is drawn by first partitioning its nodes into $K$ blocks of sizes $\{n_1,\dots,n_K\}$, where $K$ is sampled uniformly from $\{2,\dots,5\}$ and each block size is sampled uniformly from $\{20,\dots,40\}$.  Writing $z_i\in\{1,\dots,K\}$ for the block containing node $i$, each ordered pair is then assigned an edge independently,
\[
\bm{A}_{ij}\sim\mathrm{Bernoulli}\!\left(P_{z_i z_j}\right),
\]
with $P\in[0,1]^{K\times K}$ collecting the within- and between-block edge probabilities.  Because the two orientations of a pair are drawn separately, the sampled graph is directed.  We use the same setting as the reference protocol,
\[
P_{kk}=p_{\mathrm{intra}}=0.3,
\qquad
P_{kl}=p_{\mathrm{inter}}=0.05 \quad (k\neq l),
\]
which yields graphs whose edges concentrate inside blocks while leaving a sparse set of cross-block connections. We sample $200$ graphs and split them in $128/32/40$ training, validation, and test sets. Further dataset statistics can be found in Table~\ref{tab:syntheticstats}.

\begin{table}[h]
    \centering
    \caption{SBM dataset statistics.}
    \vspace{0.15cm}
    \label{tab:syntheticstats}
     \resizebox{1.0\linewidth}{!}{
    \begin{tabular}{lccccccccc}
        \midrule
        Dataset & Min. nodes & Max. nodes & Avg. nodes & Min. edges & Max. edges & Avg. edges & \#Train & \#Val & \#Test \\
        \midrule
        SBM  & 44 & 175 & 106 & 340 & 3008 & 1426 & 128 & 32 & 40 \\
        \bottomrule
    \end{tabular}}
\end{table}

\subsection{Twitter}\label[Appendix]{app.twitter_datasets}
To construct the Twitter dataset, we first remove self-loops and collapse duplicate edges in the source network. We then partition the source nodes into disjoint training, validation, and test pools before subgraph extraction, ensuring that no account appears in more than one split. Within each pool, candidate induced subgraphs are obtained through breadth-first expansion over both incoming and outgoing neighborhoods, after which only the largest weakly connected component is retained. We keep candidates containing between 64 and 80 nodes, with at least as many edges as nodes and edge reciprocity no greater than $0.42$. To limit redundancy within each split, we further discard a candidate if the Jaccard similarity between its node set and that of any previously retained graph exceeds $0.5$. Statistics of the resulting dataset are reported in \cref{tab:twitter-dataset-statistics}. 

This procedure yields 4{,}067 connected graphs with an overall mean directed density of $0.118$, where density is computed as $|\calE|/(|\calV|(|\calV|-1))$. Split-wise and aggregate statistics are reported in \cref{tab:twitter-dataset-statistics}. Twitter graphs exhibit heterogeneous connectivity on both sides of a directed edge: the mean in- and out-degree Gini coefficients are approximately $0.52$ and $0.48$, reflecting variations in both account popularity and following activity. Their mean reciprocity is $0.357$, compared with $0.166$ under an in/out-degree-preserving null model, indicating that mutual following remains selective but occurs more frequently than expected from the degree sequences alone.

Every retained graph contains at least one directed cycle, and approximately $74\%$ of nodes participate in directed cycles. The graphs also exhibit substantial local clustering and community organization, with mean directed clustering and Louvain modularity of $0.397$ and $0.301$, respectively. Together, these statistics characterize a naturally evolving social-network regime in which asymmetric connectivity, selective reciprocity, directed cycles, degree heterogeneity, and community structure coexist. Since the processed graph representation contains neither node features nor categorical node labels, Twitter is used exclusively as a structure-only benchmark, with generation and evaluation performed on binary directed adjacency matrices.

\begin{table}[h]
\centering
\caption{Twitter dataset statistics. It contains much denser graphs compared with existing datasets TPU Tiles and Visual Genome.}
\label{tab:twitter-dataset-statistics}
\small
\resizebox{\linewidth}{!}{
\begin{tabular}{lrrrrrrr}
\midrule
Dataset
& Min.\ nodes
& Max.\ nodes
& Avg.\ nodes
& Min.\ edges
& Max.\ edges
& Avg.\ edges
& \# Graphs \\
\midrule
Overall & 64 & 80 & 71.4 & 108 & 2{,}152 & 591.3 & 4{,}067 \\
\midrule
Train      & 64 & 80 & 71.4 & 108 & 1{,}959 & 545.8 & 3{,}447 \\
Validation & 64 & 80 & 70.9 & 157 & 1{,}776 & 682.5 & 198 \\
Test       & 64 & 80 & 71.6 & 125 & 2{,}152 & 920.5 & 422 \\
\bottomrule
\end{tabular}
}
\end{table}

\subsection{Circuit}\label[Appendix]{app.Circuit_datasets}

Each node represents a primitive cell instance, and its categorical label specifies one of seven cell types appearing in the retained graphs: lookup table (LUT), flip-flop (FD), multiplexer (MUX), constant driver (CONST), carry-chain cell (CARRY), memory cell (MEMORY), or input/output buffer (IOBUF). A directed edge $u\rightarrow v$ indicates that a signal produced by cell $u$ drives an input of cell $v$; hence, edge orientation follows signal flow from driver to load. Parallel pin-level connections between the same pair of cells are represented by a single binary edge, and self-loops are removed. All graphs in the Circuits dataset are weakly connected and sparse, with a mean directed density of $0.044$. The node-label distribution is highly imbalanced, with LUT and FD cells together accounting for $91.6\%$ of all nodes. Split-wise and aggregate dataset statistics are reported in \cref{tab:circuit-dataset-statistics}.

Circuit graphs represent engineered systems whose connectivity is constrained by hardware functionality. In particular, a cell's fan-in is limited by its number of input pins, whereas its output may drive many downstream cells, resulting in substantially greater out-degree than in-degree heterogeneity. Directed dependencies are rarely mutual, with an average reciprocity of only $0.012$. Approximately $42.7\%$ of the graphs are acyclic; in every remaining graph, each directed cycle contains at least one state-holding cell, reflecting the absence of purely combinational feedback loops in synchronous circuit designs. Connectivity is also strongly conditioned on cell type: only 20 of the 49 possible ordered pairs of the seven observed categories occur, and five pairs (i.e., FD$\rightarrow$LUT, LUT$\rightarrow$LUT, LUT$\rightarrow$FD, CARRY$\rightarrow$LUT, and LUT$\rightarrow$CARRY), accounting for approximately $86.4\%$ of all edges. Circuits therefore provides a benchmark for modeling constrained and task-specific directed connectivity, where successful generation requires capturing asymmetric signal dependencies and label-conditioned wiring rules rather than only matching aggregate graph statistics.

\begin{table}[!h]
\centering
\caption{Circuits dataset statistics.}
\label{tab:circuit-dataset-statistics}
\small
\resizebox{\linewidth}{!}{
\begin{tabular}{lrrrrrrr}
\midrule
Dataset
& Min.\ nodes
& Max.\ nodes
& Avg.\ nodes
& Min.\ edges
& Max.\ edges
& Avg.\ edges
& \# Graphs \\
\midrule
Overall & 30 & 60 & 44.1 & 40 & 193 & 80.5 & 318 \\
\midrule
Train      & 30 & 60 & 43.8 & 40 & 193 & 81.7 & 251 \\
Validation & 30 & 59 & 44.7 & 44 & 139 & 73.4 & 41  \\
Test       & 30 & 60 & 45.3 & 44 & 165 & 80.6 & 26  \\
\bottomrule
\end{tabular}
}
\end{table}

\subsection{TPU Tiles~\citep{google2023tpugraphs}}

Each graph in this collection is a directed acyclic computation graph extracted from a machine learning workload, where nodes are tensor operations and edges carry the data dependencies between them; a single graph corresponds to one unit of computation, such as a training epoch or an inference pass. A datapoint pairs such a graph with a compiler configuration and records the measured execution time obtained under that configuration, giving $6{,}301$ datapoints in total.  The underlying workloads span a range of architecture families, including ResNets, EfficientNets, Masked R-CNNs, and Transformers, so the collection exhibits considerable structural heterogeneity: graph size ranges from a pair of operations to roughly $400$ nodes, while the average graph has only $41$ nodes and $43$ edges.  This combination of directedness, acyclicity, and heavy-tailed size makes it a useful test of whether a generator can capture engineered functional dependencies rather than the homogeneous structure of synthetic families.

We adopt the preprocessing pipeline and predefined training/validation/test split of \citet{li2025layerdag}. The original collection has a highly dispersed size distribution, with graph sizes ranging from $2$ to $394$ nodes. Although $90\%$ of the graphs contain at most $80$ nodes, the sparsely populated upper tail extends considerably further, with $1\%$ containing more than $113$ nodes. At the other extreme, $26\%$ of the graphs contain fewer than $20$ nodes and $8\%$ contain fewer than $10$. These small graphs provide limited structural information and can yield less stable spectral, wavelet, and clustering statistics. We therefore retain graphs containing between $30$ and $80$ nodes. The filter is applied independently within each split, preserving the original split assignments. This restriction removes both size extremes and confines graph sizes to a range spanning less than a factor of three, allowing the evaluation to focus on structural fidelity within a representative size regime rather than on behavior at extreme graph sizes. The resulting subset contains $3{,}183$ graphs, comprising $2{,}539$ training, $318$ validation, and $326$ test graphs, with averages of $50.1$ nodes and $53.0$ edges per graph. All methods are trained and evaluated on the same filtered subset. Per-split statistics are given in Table~\ref{tab:tpustats}.

\begin{table}[ht]
    \centering
    \caption{TPU Tiles dataset statistics.}
    \vspace{0.15cm}
    \label{tab:tpustats}
     \resizebox{1.0\linewidth}{!}{
    \begin{tabular}{lccccccc}
        \midrule
        Dataset & Min. nodes & Max. nodes & Avg. nodes & Min. edges & Max. edges & Avg. edges & \# Graphs \\
        \midrule
        Overall     & 30 & 80 & 50  & 29 & 116 &53  & 3183 \\
        \midrule
        Train       & 30 & 80 & 50 & 29 & 116 & 53.0 & 2539 \\ 
        Validation  & 30 & 80 & 52 & 29 & 102 & 55.2 & 318 \\ 
        Test        & 30 & 80 & 48 & 29 & 102 & 50.9 & 326 \\
        \bottomrule
    \end{tabular}}
\end{table}

\subsection{Visual Genome~\citep{krishna2017visualgenome}}
\label{app:visualgenome}

Visual Genome annotates over $100{,}000$ images with objects, their attributes, and the relationships between them, and is a standard benchmark for scene graph generation, visual question answering, and grounded language understanding.  Each annotated scene induces a directed graph in which the three annotation types become typed nodes, labelled by their text, and the visual semantics fix the orientation of every edge: objects point to their attributes and to the relationships they participate in, relationships point back to the objects they take as arguments, and attributes are sinks.  The resulting graphs are therefore directed and partly cyclic, with $119$ of them ($\sim$37.5\%) containing at least one cycle, which makes the dataset a natural test case for methods that cannot assume a topological ordering.

We take the graph construction, vocabulary filtering, and splits directly from \citet{carballocastro2026directo}, so that our results are comparable to those reported there.  In brief, that pipeline retains scenes with $20$--$40$ nodes and at least $25$ edges, and restricts node labels to the $20$ most frequent objects, relationships, and attributes, giving $60$ node classes in total.  Split statistics are reported in Table~\ref{tab:visgenstats}; the full label vocabulary is given in the original work.

\begin{table}[h]
    \centering
    \caption{Visual Genome dataset statistics.}
    \vspace{0.15cm}
    \label{tab:visgenstats}
    \resizebox{1.0\linewidth}{!}{
    \begin{tabular}{lccccccccc}
        \midrule
        Dataset & Min. nodes & Max. nodes & Avg. nodes & Min. edges & Max. edges & Avg. edges & \# Graphs \\
        \midrule
        Overall & 21 & 40 & 34 & 25 & 47 & 28 & 317 \\
        \midrule
        Train & 21 & 40 & 33 & 25 & 47 & 28 & 203  \\
        Validation  & 24 & 40 & 34 & 25 & 40 & 28 &  51 \\
        Test & 21 & 40 & 34 & 25 & 36 & 28 &  63 \\
        \bottomrule
    \end{tabular}}
\end{table}

\section{Further details on evaluation metrics}\label[Appendix]{app.eval.metric}

\subsection{Validity, uniqueness, and novelty.} Validity is evaluated using dataset-specific criteria. For synthetic datasets, we test whether generated graphs are consistent with the underlying generative distribution; for DAG datasets, validity requires acyclicity. Accordingly, on TPU Tiles, we report the fraction of acyclic samples, whereas on Visual Genome, we verify typed structural constraints requiring edges to point from objects to attributes or relationship nodes, and from relationship nodes to objects. We do not define validity for Circuits or Twitter because neither representation admits a complete binary criterion: circuit functionality cannot be established from extracted cell-level subgraphs, while social-network plausibility is inherently distributional. Their fidelity is instead evaluated through distributional graph statistics. For all datasets, uniqueness denotes the fraction of generated graphs that are distinct up to isomorphism, while novelty denotes the fraction not present in the training set. The V.U.N. ratio measures the proportion of samples that are simultaneously valid, unique, and novel and is reported only for datasets with a defined validity criterion.

\subsubsection{Validity metrics}

\paragraph{Stochastic block model} Following \citet{carballocastro2026directo}, we assess whether a generated graph is consistent with the target directed SBM through its blockwise connection probabilities. Given a graph $G=(\mathcal V,\mathcal E)$ with adjacency matrix $\mathbf A$, we first infer a partition into $B$ non-empty blocks using the minimum-description-length method of \citet{peixoto2014efficisbm}:  
\begin{align*}
\widehat z = \arg\min_{z:\mathcal V\rightarrow{1,\ldots,B}} \calL(G,z),  
\end{align*}
where $\mathcal L(G,z)$ denotes the description length. The inferred partition can be further refined through $T$ MCMC steps. Let $N_i$ be the number of nodes assigned to block $i$ and $E_{ij}$ the number of directed edges from block $i$ to block $j$. The empirical within- and between-block connection probabilities are
\begin{equation}
\widehat p_{\mathrm{intra},i}
=
\frac{E_{ii}}{N_i(N_i-1)+\varepsilon},
\qquad
\widehat p_{\mathrm{inter},ij}
=
\frac{E_{ij}}{N_iN_j+\varepsilon},
\quad i\neq j,
\end{equation}
where $\varepsilon=10^{-6}$ ensures numerical stability. We compare these estimates with the target probabilities $p_{\mathrm{intra}}$ and $p_{\mathrm{inter}}$ using 
\begin{equation}
W_{ii}
=
\frac{
(\widehat p_{\mathrm{intra},i}-p_{\mathrm{intra}})^2
}{
\widehat p_{\mathrm{intra},i}
(1-\widehat p_{\mathrm{intra},i})+\varepsilon
},
\qquad
W_{ij}
=
\frac{
(\widehat p_{\mathrm{inter},ij}-p_{\mathrm{inter}})^2
}{
\widehat p_{\mathrm{inter},ij}
(1-\widehat p_{\mathrm{inter},ij})+\varepsilon
},
\quad i\neq j.
\end{equation}
Each statistic is converted into a block-pair compatibility value using a chi-squared reference distribution with one degree of freedom,
\begin{equation}
q_{ij}
=
1-F_{\chi^2_1}(W_{ij}),
\qquad
s_{\mathrm{SBM}}(G)
=
\frac{1}{B^2}
\sum_{i=1}^{B}\sum_{j=1}^{B}q_{ij}.
\end{equation}
The resulting graph-level score averages agreement across all ordered block pairs, with larger values indicating closer conformity to the prescribed SBM parameters.

\paragraph{TPU Tiles}
For the TPU Tiles dataset, validity is determined by acyclicity, consistent with its representation of computational pipelines as directed acyclic graphs (DAGs). We report the proportion of generated graphs containing no directed cycles.

\paragraph{Visual Genome} 
Following \citet{carballocastro2026directo}, we evaluate the typed structural constraints of Visual Genome scene graphs, which contain object, relationship, and attribute nodes. Only three directed edge types are permitted: object-to-relationship, relationship-to-object, and object-to-attribute. We adapt the original binary graph-level criterion by evaluating these constraints at the edge level. Under the binary formulation, a single violating edge makes the entire graph invalid; this all-or-nothing criterion becomes increasingly stringent as the number of edges grows and provides little separation among models on graphs of this size. We therefore define
\begin{equation}
\operatorname{Validity}(G)
=
\frac{\text{number of type-valid directed edges in }G}
{|\calE|}.
\end{equation}
where $|\calE|$ denotes the total number of directed edges in $G$. For each sampling run, we average this score over the generated graphs and report its mean and standard deviation across five runs. Higher values indicate greater compliance with the scene-graph typing rules, and every reference graph has an edge-validity score of one by construction.

\subsection{Structural distributional alignment}
To assess how closely generated graphs match the data distribution, we compute the Maximum Mean Discrepancy (MMD) \citep{martinkus2022spectre,thompson2022evaluation} between generated and test graphs over in- and out-degree distributions, directed clustering coefficients, spectral, and wavelet features. Spectral features are computed using Chung's directed Laplacian \citep{chung2005laplacians}, which is computationally simpler than magnetic-Laplacian alternatives. For each descriptor, we normalize the generated--test MMD by the corresponding train--test MMD and average the resulting ratios across descriptors. Lower values indicate better alignment. Descriptor-level results are provided below.

\paragraph{Maximum mean discrepancy metrics}\label[Appendix]{app.mmd_metrics}
Following \citet{martinkus2022spectre} and its extension to directed graphs by \citet{carballocastro2026directo}, we evaluate structural distributional alignment using Maximum Mean Discrepancy (MMD) over five descriptors: out-degree, in-degree, directed clustering, spectrum, and wavelet features. In- and out-degree distributions are evaluated separately to capture asymmetric connectivity. For the clustering descriptor, we compute the distribution of local directed clustering coefficients using the \texttt{clustering()} function in \texttt{NetworkX}, which supports directed graphs. For the spectral and wavelet descriptors, we replace the conventional undirected Laplacian with Chung's directed Laplacian \citep{chung2005laplacians}, which derives spectral information from directed random-walk dynamics while remaining computationally simpler than magnetic-Laplacian alternatives. 

For each descriptor, we compute the MMD between the generated and test sets and normalize it by the corresponding MMD between the training and test sets. The reported average ratio is the mean of the five normalized scores. Lower values indicate closer structural alignment, while a ratio of one corresponds to the train--test baseline. We omit the orbit-count metric used in undirected benchmarks because its standard ORCA-based implementation does not support directed graphlets, and graph symmetrization would discard edge-direction information.

\subsection{Joint node-edge distributional metrics}\label[Appendix]{app.joint_nodoedge_metrics}

Structural statistics alone cannot determine whether generated graphs preserve the dependence between categorical node labels and directed connectivity. For datasets with categorical node labels, we extend the evaluation beyond structural statistics to assess both node-label marginals and their dependence on directed connectivity. First, node-label MMD compares graph-level label histograms to assess alignment of marginal categorical proportions. Second, we extract dataset-specific semantic tuples, such as $(\text{object},\text{relation},\text{object})$, and report precision and recall to measure the accuracy and coverage of typed connectivity patterns relative to the test set. We additionally compute Triplet MMD between tuple-frequency histograms to quantify their distributional agreement. Finally, following \citet{thompson2022evaluation}, we compute RBF-kernel MMD over GNN-derived graph embeddings that jointly encode node labels and structure, thereby measuring global attributed-graph similarity. These metrics are reported only for datasets with categorical node labels.

\paragraph{Node-label distribution MMD.}
For each graph, we construct a normalized histogram over node categories and compute the MMD between the histograms of generated and test graphs. This metric measures whether the generated graphs reproduce the distribution of node-type proportions in the reference data.

\paragraph{Multiplicity-aware triplet-based precision and recall.}
We build on the Triplet Precision and Recall metrics of \citet{carballocastro2026directo}, extending them to account for semantic-tuple multiplicity. For Visual Genome, we extract an $(\text{object},\text{attribute})$ pair for each attribute node and a $(\text{subject},\text{relationship},\text{object})$ triplet for each relationship node. For TPU Tiles and Circuits, each directed edge $u\rightarrow v$ contributes $(\ell(u),\texttt{connects},\ell(v))$. For each tuple type, we compute its average occurrence rate per graph in the generated and test sets. Its matched contribution is the smaller of these two rates, and summing these contributions over all tuple types gives the total number of matched tuple occurrences. Precision and recall are then defined as
\begin{align*}
\operatorname{Precision}
= \frac{\text{matched tuple occurrences}}
{\text{generated tuple occurrences}},
\qquad
\operatorname{Recall}
= \frac{\text{matched tuple occurrences}}
{\text{test tuple occurrences}}.  
\end{align*}

Higher precision indicates that a larger proportion of generated tuple occurrences is matched by the test graphs, whereas higher recall indicates that a larger proportion of test tuple occurrences is recovered by the generated graphs. Both metrics range from zero to one, with higher values indicating better agreement, and both equal one only when every tuple type has the same average per-graph occurrence rate in the generated and test graphs.

\paragraph{Triplet MMD.}
Using the dataset-specific semantic tuples defined above, we construct tuple-frequency representations for the generated and test graphs and compute the MMD between the resulting distributions. This metric evaluates whether label-conditioned connectivity patterns occur with similar relative frequencies in the generated and reference graphs, complementing triplet-based precision and recall, which primarily measure tuple overlap.

\paragraph{Embedding-based RBF MMD.}
Following \citet{thompson2022evaluation}, we obtain graph-level representations using a fixed GNN encoder that incorporates both node labels and directed structure. We then compute MMD with an RBF kernel between the embeddings of generated and test graphs, providing a global measure of joint node--edge distributional alignment.

\section{Experimental details}

\subsection{Details on baselines}
\label{app:expbaselines}

\paragraph{Maximum likelihood estimation (MLE).}
Following \citet{carballocastro2026directo}, we include a simple count-based baseline that fits three empirical distributions to the training split and samples from them independently.  The first is the distribution over graph size, the second is the marginal over node classes, and the third is the edge-type distribution conditioned on the classes of the two endpoints. A sample is then drawn by first choosing $n$, then drawing each node class independently, and finally drawing an edge type for every ordered pair conditioned on its endpoint classes.  The baseline therefore reproduces first-order statistics by construction while ignoring any higher-order structure.

\paragraph{D-VAE~\citep{zhang2019dvae}.}
D-VAE encodes and decodes DAGs with an asynchronous message-passing scheme: rather than updating all nodes simultaneously as a standard GNN would, it processes them one at a time in topological order, so that the representation reflects the computation the graph performs and not only its connectivity. Generation mirrors this order, adding nodes sequentially with all new edges oriented toward the node just introduced, which guarantees acyclicity; inputs must be topologically sorted beforehand.  We follow the setup of \citet{carballocastro2026directo}, using the authors' \href{https://github.com/muhanzhang/D-VAE}{released implementation} with the hyperparameters recommended in the original paper and converting generated graphs back to our evaluation format.  On TPU Tiles, however, the size of the dataset and of its largest graphs made training impractical within our compute budget.

\paragraph{LayerDAG~\citep{li2025layerdag}.}
LayerDAG recasts a DAG as a sequence of bipartite graphs, giving a layer-wise tokenization over which generation proceeds autoregressively; within each step, a diffusion process models the dependencies among nodes that are mutually incomparable, that is, not joined by any directed path.  As with D-VAE, inputs must be topologically sorted.  We use the authors' \href{https://github.com/Graph-COM/LayerDAG}{released implementation} with its default configurations, including the per-component epoch budgets. 

\paragraph{DiGress~\citep{vignac2023digress}.}
We benchmark against the directed variant of \textsc{DiGress} implemented by \citet{carballocastro2026directo}, obtained by removing the symmetrization operations so that the two orientations of a node pair are modeled separately.  In the terminology of that work, this configuration is \textsc{Directo}-DD without dual attention and with Laplacian positional encodings.

\paragraph{DeFoG~\citep{qin2025defog}.}
Likewise, we use the directed variant of \textsc{DeFoG} from \citet{carballocastro2026directo}, again obtained by dropping the symmetrization operations.  This configuration corresponds to \textsc{Directo} without dual attention and with RRWP positional encoding.

\paragraph{Directo~\citep{carballocastro2026directo}.}
\textsc{Directo} is the closest baseline to our setting: it applies discrete flow matching directly to general directed graphs, combining dual attention over incoming and outgoing neighbourhoods with asymmetric positional encodings.  We benchmark against both reported variants, \textsc{Directo}-DD and \textsc{Directo}, and for each we run the two positional encodings used in the original paper, RRWP and the magnetic Laplacian (MagLap).  All four configurations use the authors' released implementation and their published hyperparameters.

\subsection{Implementation details}\label[Appendix]{app:implementation}

\paragraph{Network architectures.} 
The joint velocity field \( v_\phi \) is parameterized by a residual MLP with a backbone shared across its four velocity components. The shared backbone consists of an input projection followed by four pre-normalized, two-layer residual blocks with ReLU activations and a hidden dimension of 1{,}024. For datasets with node labels, the structure-conditioned predictor $f_\psi$ is implemented as a direction-aware GNN with four residual message-passing layers and a hidden dimension of 256. These architectural settings are used across all datasets.

\paragraph{Choice of spectral dimension.} For each dataset, we select \(k\) by minimizing the validation-set average structural MMD ratio over even-valued candidates not exceeding the dataset's maximum node count \(N_{\max}\). Here \(k\) specifies the maximum number of retained spectral components. For an \(n\)-node graph with \(n<k\), we retain the largest even number of components not exceeding \(n\) and zero-pad the resulting channels for batching. On Circuits, we evaluate \(k\in\{4,8,16,32\}\); \(k=32\) achieves the lowest average structural MMD ratio. Results for the other datasets are reported in \cref{tab:spectral_rank_selection}. Across the evaluated settings on Circuits, Visual Genome, and SBM, increasing \(k\) lowers the average structural MMD ratio, suggesting that retaining more spectral information improves structural distribution matching.


\begin{table}[h]
\centering
\caption{
Dataset-specific spectral-rank selection using the validation average structural MMD ratio ($\downarrow$). Selected settings are bolded.
}
\label{tab:spectral_rank_selection}

\small
\setlength{\tabcolsep}{4pt}
\renewcommand{\arraystretch}{0.92}
\setlength{\aboverulesep}{0.2ex}
\setlength{\belowrulesep}{0.2ex}

\begin{tabular}{@{}lcc@{}}
\midrule
\textbf{Dataset}
& \textbf{Rank} $k$
& \textbf{Avg. MMD ratio} $\downarrow$ \\
\midrule

\multirow{4}{*}{Circuit}
& $4$  & $3.748 \pm 0.098$ \\
& $8$  & $2.609 \pm 0.152$ \\
& $16$ & $2.453 \pm 0.210$ \\
& $\mathbf{32}$ & $\mathbf{1.766 \pm 0.130}$ \\
\midrule
\multirow{4}{*}{Visual Genome}
& $4$  & $38.979 \pm 0.286$ \\
& $8$  & $31.018 \pm 1.288$ \\
& $16$ & $10.723 \pm 1.230$ \\
& $\mathbf{32}$ & $\mathbf{3.206 \pm 0.378}$ \\
\midrule
\multirow{5}{*}{SBM}
& $16$  & $15.771 \pm 0.604$ \\
& $32$  & $13.680 \pm 0.275$ \\
& $64$  & $5.627 \pm 0.848$ \\
& $\mathbf{96}$ & $\mathbf{2.694 \pm 0.378}$ \\
\bottomrule
\end{tabular}
\end{table}

\paragraph{Optimization and generation.}
We train the structural model \(v_\phi\) and, where applicable, the label predictor \(f_\psi\) separately using AdamW with a weight decay of \(10^{-4}\). Following the training protocol of \citet{carballocastro2026directo}, we train \(v_\phi\) for
\(10{,}000\) epochs on each dataset and retain the checkpoint with the lowest average MMD ratio on the validation set. For \(f_\psi\), we use a fixed \(2{,}000\) training epochs, which was sufficient for the training objective to stabilize, and use the final checkpoint for evaluation. During its training, the conditioning adjacency is perturbed through edge-count-preserving random rewiring to reduce sensitivity to discrepancies between ground-truth structures used for training and generated structures encountered at inference. For endpoint construction, we use a dataset-specific source template: a directed path (Chain), a directed cycle (Cycle), or an augmented directed cycle (Aug.\ Cycle) with additional forward and selected reverse edges. In every case, the four source channels are obtained by decomposing the same perturbed source matrix and therefore jointly correspond to a single matrix at $t=0$. The ablation in \cref{tab:source_construction_ablation} indicates that this shared-matrix coupling is the primary contributor, while the particular template topology has a smaller additional effect. Dataset-specific hyperparameters choosen are summarized in \cref{tab:implementation_hyperparameters}.


\begin{table*}[!h]
\centering
\caption{Dataset-specific hyperparameters for PolarFlow.
Here, \(k\) denotes the spectral truncation dimension, \(\varepsilon\) the diagonal regularization strength applied to the target adjacency, \(\sigma\) the scale of the Gaussian noise added to the source template, and \(N_{\mathrm{Euler}}\) the number of forward Euler steps used to integrate the generative ODE during sampling. Dashes indicate datasets without node labels.}
\label{tab:implementation_hyperparameters}

\scriptsize
\setlength{\tabcolsep}{4pt}
\renewcommand{\arraystretch}{1.05}

\begin{tabular}{@{}lcccccccccc@{}}
\midrule
& \multicolumn{4}{c}{Endpoint construction}
& \multicolumn{2}{c}{Structural model ($v_{\phi}$)}
& \multicolumn{1}{c}{Sampling}
& \multicolumn{2}{c}{Label model ($f_{\psi}$)} \\
\cmidrule(lr){2-5}
\cmidrule(lr){6-7}
\cmidrule(lr){8-8}
\cmidrule(lr){9-10}

Dataset
& $k$
& $\varepsilon$
& $\sigma$
& source template
& learning rate
& batch size
& \(N_{\mathrm{Euler}}\)
& learning rate
& batch size \\
\midrule

SBM
& 96
& 0
& 0.1
& Aug. cycle
& 3e-4
& 16
& 200
& --
& --
\\

Twitter
& 64
& 0.05
& 0.3
& Aug. cycle
& 3e-4
& 32
& 200
& --
& --
\\

Circuits
& 32
& 0.1
& 0.1
& Chain
& 3e-4
& 32
& 200
& 1e-3
& 64
&  \\

Visual Genome
& 32
& 0.1
& 0.1
& Cycle
& 3e-4
& 32
& 200
& 1e-3
& 64
&  \\

TPU Tiles
& 32
& 0.3
& 0.1
& Chain
& 3e-4
& 32
& 200
& 1e-3
& 64
& \\

\bottomrule
\end{tabular}
\end{table*}

\subsection{Ablation studies}\label[Appendix]{app.ablation}

We conduct controlled ablations of two design choices affecting PolarFlow's stability and coherence: (i) applying a small diagonal perturbation to the data endpoint before polar--spectral decomposition to mitigate rank and spectral degeneracies, and (ii) jointly constructing the four source channels from a shared noisy matrix to preserve coupling among the polar factors and their spectral coordinates.
\paragraph{Effect of diagonal endpoint regularization.}\label[Appendix]{sec:regularization_ablation}
Motivated by the stability analysis \cref{thm:stability}, we instantiate the endpoint perturbation as the diagonal shift  $\bA_\varepsilon= \bA + \varepsilon \bI_N$ and evaluate $\varepsilon\in\{0,0.1,0.2,0.3\}$, where $\varepsilon=0$ corresponds to the unregularized endpoint. The shift is applied before the polar--spectral decomposition and does not introduce self-loops, since diagonal entries are excluded during discretization. Rank deficiency and spectral degeneracies, including repeated eigenvalues and weakly separated spectral modes, can make the polar--spectral endpoint representation non-unique or sensitive to small input perturbations. The diagonal shift may alleviate these issues by increasing the rank margin $\sigma_{\min}(\bA_\varepsilon)$ or widening the relevant spectral gaps of $\bP$ and $\bU$, yielding better-conditioned endpoints and, consequently, more stable conditional paths and velocity targets. As shown in \Cref{tab:diagonal_shift_ablation}, every tested nonzero shift achieves a lower mean structural average ratio than the unregularized setting on TPU Tiles and Circuit, while $\varepsilon=0.1$ yields the best result on Twitter. Overall, these results support the practical value of stability-motivated endpoint regularization, indicating that an appropriately chosen regularization strength can improve generation quality.

\begin{table}[!h]
    \centering
    \caption{\small Effect of the diagonal endpoint shift on structural generation quality via average ratio ($\downarrow$; mean $\pm$ standard deviation). $\varepsilon=0$ denotes the unregularized setting, and the best results are bolded.
    }
    \label{tab:diagonal_shift_ablation}
    \setlength{\tabcolsep}{4pt}
    \renewcommand{\arraystretch}{0.9}
    \small
    \begin{tabular}{lcccc}
        \midrule
        Dataset
        & $\varepsilon=0$
        & $\varepsilon=0.1$
        & $\varepsilon=0.2$
        & $\varepsilon=0.3$ \\
        \midrule
        TPU Tiles
        & $16.44\pm0.36$
        & $\mathbf{11.67\pm0.87}$
        & $14.68\pm0.81$
        & $11.81\pm0.38$ \\

        Circuit
        & $3.73\pm0.11$
        & $1.89\pm0.22$
        & $1.88\pm0.30$
        & $\mathbf{1.81\pm0.20}$ \\

        Twitter
        & $1.61\pm0.02$
        & $\mathbf{1.33\pm0.04}$
        & $2.26\pm0.19$
        & $1.64\pm0.10$ \\
        \bottomrule
    \end{tabular}
\end{table}

\paragraph{Effect of coupled source construction.}
To isolate the contribution of the coupled source distribution, we replace the default source state with an independent prior. PolarFlow constructs all four source channels jointly by applying the polar--spectral decomposition to a single source matrix, thereby preserving their cross-channel coupling. In the independent variant, no matrix is decomposed at $t=0$; instead, the four channels are sampled separately as
\begin{align*}
\bQ_{u,0},\bQ_{p,0}
\overset{\mathrm{ind.}}{\sim}
\operatorname{Unif}\left(\St(N,k)\right),
\quad
\btheta_0^{(k)}\sim\operatorname{Unif}\left((-\pi,\pi]^k\right),
\quad
\blambda_{0}^{(k)} \sim \rho_{\lambda}~\text{with}~\operatorname{supp}(\rho_{\lambda})\subseteq\Real_{+}^k,
\end{align*}
where $\bQ_{u,0}$, $\bQ_{p,0}$, $\btheta_0^{(k)}$, and $\blambda_{0}^{(k)}$ are mutually independent. Here, $\rho_{\lambda}$ denotes the positive prior induced by applying the elementwise transformation $x\mapsto\log(1+\exp(x))$ to a standard Gaussian vector $\bx\sim\mathcal{N}(\mathbf{0},\bI_k)$. This construction respects the geometry of each channel while removing the coupling induced by their derivation from a common source matrix.

As shown in \cref{tab:source_construction_ablation}, the independent prior yields higher average ratios on both Circuit and Visual Genome, indicating degraded structural generation quality. This result highlights the benefit of jointly constructing the four source channels through the polar--spectral decomposition of a common source matrix, thereby preserving their cross-channel consistency. To examine whether the improvement instead arises from the structural prior encoded by the template, we compare the structured-template coupled source with a Gaussian-only coupled variant. For the latter, we set the deterministic template to $\bA_0=\mathbf{0}_{N\times N}$, such that the source matrix reduces to $\bM_0=\sigma\bF$. This removes the template-specific structural information while preserving the joint construction of all four channels from the same source matrix. The Gaussian-only variant yields results close to those of the structured-template source and remains substantially better than the independent prior. Taken together, these results suggest that the improvement is driven primarily by the joint construction of the four channels from a shared source matrix, rather than by the particular structure encoded in the template.

\begin{table}[h]
    \centering
    \caption{
        Effect of source coupling and template choice on structural generation quality, measured by average ratio ($\downarrow$; mean $\pm$ standard deviation). Best results are shown in bold.
    }
    \label{tab:source_construction_ablation}

    \begin{minipage}{\linewidth}
        \centering
        \footnotesize
        \setlength{\tabcolsep}{6pt}
        \renewcommand{\arraystretch}{0.8}

        \begin{tabular}{@{}lccc@{}}
            \midrule
            Dataset
            & Independent
            & \multicolumn{2}{c}{Coupled} \\
            \cmidrule(lr){3-4}
            & & Gaussian only ($\bA_0=\mathbf{0}$)
            & Structured template \\
            \midrule
            Circuit
            & $2.59 \pm 0.21$
            & $1.88 \pm 0.36$
            & $\mathbf{1.69 \pm 0.22}$ (Chain) \\
            Visual Genome
            & $6.23 \pm 0.67$
            & $3.87 \pm 0.59$
            & $\mathbf{3.32 \pm 0.53}$ (Cycle) \\
            \bottomrule
        \end{tabular}

        \par\vspace{3pt}
        {\scriptsize\raggedright
        \textit{Note:} The structured variant uses a directed chain for
        Circuit to reflect its ordered signal-flow dependencies and a directed
        cycle for Visual Genome, whose relations are not assumed to follow an
        acyclic ordering. 
        \par}
    \end{minipage}
\end{table}

\subsection{Node-label generation}\label[Appendix]{app.label_decoding}
In directed graphs, categorical node labels often encode structural or functional roles associated with patterns of incoming and outgoing connectivity. To capture this dependence, PolarFlow generates node labels in a second stage conditioned on the graph structure. Given an adjacency matrix $\bA\in\{0,1\}^{N\times N}$ and node labels $\by\in\calC^N$, this stage models the conditional distribution $p_\psi(\by\mid\bA)$. During training, $\bA$ is a data adjacency matrix; at inference, it is replaced by the generated adjacency matrix $\widehat{\bA}$, which remains fixed throughout label generation.

We introduce a label time $\tau\in[0,1]$, independent of the structural-flow time $t$, and augment the label space $\calC$ with a mask state $\mathsf{M}$, giving $\overline{\calC}=\calC\cup\{\mathsf{M}\}$. For each data label vector $\by$, we define a conditional probability path $q_\tau(\by_\tau\mid\by)$ from the all-mask state $\by_0=(\mathsf{M},\ldots,\mathsf{M})$ to the data label vector $\by_1=\by$:
\vspace{-2mm}
\begin{align*}
q_\tau(\by_\tau\mid\by)= \prod_{i=1}^{N} 
    \left[
        (1-\tau)\mathbf{1}\{y_{\tau,i}=\mathsf{M}\}
        + \tau\mathbf{1}\{y_{\tau,i}=y_i\}
    \right].
\end{align*}
Thus, at time $\tau$, each node remains masked with probability $1-\tau$ and takes it data label with probability $\tau$. These marginals are coupled monotonically across time: each node switches from $\mathsf{M}$ to its endpoint label once and remains fixed thereafter.


Let $ \calI_\tau = \{i \in \calV: y_{\tau,i}=\mathsf{M}\}$ denote the set of masked nodes. Given the partially revealed label vector $\by_\tau$, the fixed adjacency matrix $\bA$, a direction-aware graph neural network $f_\psi$ uses separate incoming and outgoing aggregations to predict a distribution $\pi_\psi^i(\cdot \mid \by_\tau,\bA,\tau)$ over $\calC$ for each masked node $i \in \calI_\tau$. 
The model is trained by minimizing cross-entropy over masked nodes:
\begin{equation}
    \calL_{\mathrm{label}}(\psi)
    =
    -\mathbb E
    \left[\frac{1}{|\calI_{\tau}|}\sum_{i\in\calI_{\tau}}
        \log
        \pi_\psi^i
        \left(
            y_i\mid\by_\tau,\bA,\tau
        \right)
    \right].
    \label{eq:label_loss}
\end{equation}
where the expectation is take over $(\bA,\by)\sim p_{\mathrm{data}}$,  $\tau\sim\operatorname{Unif}(0,1)$ and $\by_\tau\sim q_\tau(\cdot\mid\by)$. 

At inference, label generation starts from the all-mask state and remains conditioned on the fixed generated adjacency matrix $\widehat{\bA}$. At each iteration, $f_\psi$ predicts a distribution over the final label of every masked node. For each masked node, one label is sampled from its predicted distribution. A stochastic reveal score is then assigned to the resulting node-label pair, with higher-probability samples tending to receive higher scores. The procedure reveals the labels of the highest-scoring nodes until the next target reveal level is reached. The remaining nodes stay masked, while previously revealed labels remain fixed. The updated partially labeled state is passed back into $f_{\psi}$ to refresh the predictions for the remaining masked nodes. This process continues until all nodes have been assigned labels, yielding $\hat{\by}$.
The complete decoding procedure is provided in \cref{alg:label_decoding}.


\Cref{alg:label_decoding} details the iterative procedure for sampling node labels conditioned on a generated adjacency matrix. Decoding begins from the all-mask state and proceeds for \(K\) steps along the nominal cosine label-time grid \(\tau_k=1-\cos(\pi k/(2K))\). At iteration \(k\), \(f_\psi\) predicts a categorical label distribution \(\pi_i\) for each masked node \(i\), conditioned on the current partially labeled state and the fixed adjacency matrix. Following the Gumbel-Top-\(k\) construction \citep{kool2019stochastic}, a label \(c_i\sim\pi_i\) is sampled for each masked node and assigned the stochastic reveal score
\[
s_i
=
\log\pi_i(c_i)+\eta_k g_i,
\qquad
g_i\overset{\mathrm{i.i.d.}}{\sim}\operatorname{Gumbel}(0,1),
\qquad
\eta_k
=
\eta_0\left(1-\frac{k+1}{K}\right).
\]
Here, \(\eta_0\) denotes the base scale of the Gumbel perturbation. The next label time \(\tau_{k+1}\) determines the target number of masked nodes after iteration \(k\). Subject to revealing at least one node whenever masks remain, the procedure commits the sampled labels of the \(m\) highest-scoring masked nodes. The other masked nodes remain masked, while labels committed in previous iterations remain unchanged. At the next iteration, \(f_\psi\) is applied again to the updated partially labeled state to predict labels for the remaining masked nodes. This process terminates when no masked nodes remain, yielding \(\widehat{\by}\).

\begin{algorithm}[!h]
\caption{Structure-conditioned node-label decoding}
\label{alg:label_decoding}
\textbf{Input:} Generated adjacency matrix $\widehat{\bA}$; trained label predictor $f_\psi$; number of decoding steps $K$; base Gumbel scale $\eta_{\mathrm{base}}$\\
\textbf{Output:} Generated label vector $\widehat{\by}$
\begin{algorithmic}[1]
\State $\widehat{\by}\gets(\mathsf{M},\ldots,\mathsf{M})$
\State Set
$\tau_k\gets 1-\cos\!\left(\frac{\pi k}{2K}\right)$
for $k=0,\ldots,K$

\For{$k=0,\ldots,K-1$}
    \State $\calI\gets
    \{i\in\calV:\widehat y_i=\mathsf{M}\}$

    \If{$\calI=\varnothing$}
        \State \Return $\widehat{\by}$
    \EndIf

    \State $\pi_i\gets
    \pi_\psi^i(\cdot\mid\widehat{\by},\widehat{\bA},\tau_k)$
    for all $i\in\calI$

    \State $\eta_k\gets
    \eta_{\mathrm{base}}\left(1-\frac{k+1}{K}\right)$
    
    \For{$i\in\calI$}
        \State Sample a label $c_i\sim\pi_i$ and Gumbel noise
        $g_i\sim\operatorname{Gumbel}(0,1)$
        \State $s_i\gets\log\pi_i(c_i)+\eta_k g_i$
    \EndFor

    \State $m\gets
    \min\left\{
        |\calI|,
        \max\left\{
            1,\,
            |\calI|-\left\lfloor N(1-\tau_{k+1})\right\rfloor
        \right\}
    \right\}$
    \Comment{Number of labels revealed}

    \State $\calS\gets$ indices of the $m$ largest scores
    $\{s_i:i\in\calI\}$

    \For{$i\in\calS$}
        \State $\widehat y_i\gets c_i$
    \EndFor
\EndFor

\State \Return $\widehat{\by}$

\end{algorithmic}
\end{algorithm}

\subsection{Resources and runtime}\label[Appendix]{app.runtime}
All runtime measurements were obtained on a single NVIDIA A100-SXM4 GPU with 80,GB of memory. \Cref{tab:complete_runtime} summarizes the computational requirements of training and sampling across datasets and methods, including the number of training and generated (sampled) graphs, training time per epoch, average sampling time per graph, and training VRAM.
For completeness, we additionally include the empirical MLE baseline and, on the DAG dataset TPU Tiles, the DAG-specific D-VAE and LayerDAG baselines.

\begin{table*}[t]
    \centering
    \caption{\small
    Complete training and sampling resource comparison across datasets and methods. Training time per epoch, and
    sampling time is averaged over individually generated graphs. VRAM is measured during training.
    }
    \label{tab:complete_runtime}
    \scriptsize
    \setlength{\tabcolsep}{5pt}
    \renewcommand{\arraystretch}{0.92}
    \begin{tabular}{@{}lccccc@{}}
        \midrule
        Model
        & \# Train graphs
        & Training time (s/epoch)
        & \# Sampled
        & Sampling time (s/graph)
        & VRAM (GB) \\
        \midrule

        \multicolumn{6}{c}{\textbf{SBM}} \\
        \midrule
        MLE                           & 128 & 7.20 & 40 & 0.01 & 0.15 \\
        DiGress                       & 128 & 2.81 & 40 & 4.33 & 44.32 \\
        DeFoG                         & 128 & 2.45 & 40 & 6.04 & 47.30\\
        Directo-DD RRWP               & 128 & 4.40 & 40 & 11.27 & 74.73 \\
        Directo-DD MagLap             & 128 & 5.76 & 40 & 21.07 & 74.66 \\
        Directo RRWP                  & 128 & 4.42 & 40 & 11.20 & 66.66 \\
        Directo MagLap                & 128 & 5.88 & 40 & 21.22 & 74.94 \\
        PolarFlow                     & 128 & 2.66 & 40 & 0.13  & 2.32 \\
        
        \midrule

        \multicolumn{6}{c}{\textbf{Twitter}} \\
        \midrule
        MLE                            & 3447 & 14.15 & 422 & 2e-4 & 1.70 \\
        DiGress                        & 3447 & 21.56 & 422 & 1.65 & 8.89 \\
        DeFoG                          & 3447 & 16.85 & 422 & 1.80 & 8.96 \\
        Directo-DD RRWP                & 3447 & 29.67 & 422 & 1.69 & 18.67 \\
        Directo-DD MagLap              & 3447 & 45.19 & 422 & 3.77 & 18.63 \\
        Directo RRWP                   & 3447 & 29.67 & 422 & 3.56 & 18.66 \\
        Directo MagLap                 & 3447 & 44.63 & 422 & 7.78 & 18.63 \\
        PolarFlow                      & 3447 & 3.67  & 422 & 0.19 & 1.29 \\

        \midrule
        
        \multicolumn{6}{c}{\textbf{Circuit}} \\
        \midrule
        MLE                            & 251 & 34.8 & 26 & 1e-4 & 0.74 \\
        DiGress                        & 251 & 1.47 & 26 & 1.38 & 7.89 \\
        DeFoG                          & 251 & 1.13 & 26 & 1.91 & 8.15 \\
        Directo-DD RRWP                & 251 & 1.45 & 26 & 1.00 & 17.79 \\
        Directo-DD MagLap              & 251 & 2.09 & 26 & 2.17 & 17.79 \\
        Directo RRWP                   & 251 & 1.42 & 26 & 1.88 & 18.27 \\
        Directo MagLap                 & 251 & 2.09 & 26 & 6.66 & 18.27 \\
        PolarFlow                      & 251 & 1.10 & 26 & 0.16 & 1.03 \\

        \midrule

        \multicolumn{6}{c}{\textbf{Visual Genome}} \\
        \midrule
        MLE                            & 203 & 72   & 63 & 0.06 & 0.15 \\
        DiGress                        & 203 & 0.70 & 63 & 0.74 & 3.59 \\
        DeFoG                          & 203 & 0.54 & 63 & 0.58 & 4.04 \\
        Directo-DD RRWP                & 203 & 0.77 & 63 & 0.53 & 7.99 \\
        Directo-DD MagLap              & 203 & 1.21 & 63 & 1.44 & 8.01 \\
        Directo RRWP                   & 203 & 0.78 & 63 & 0.96 & 8.69 \\
        Directo MagLap                 & 203 & 1.16 & 63 & 2.73 & 8.72 \\
        PolarFlow                      & 203 & 0.81 & 63 & 0.04 & 1.02 \\

        \midrule
        \multicolumn{6}{c}{\textbf{TPU Tiles}} \\
        \midrule
        MLE                            & 2539 & 7.43 & 326 & 1e-4 & 0.05  \\
        D-VAE                          & 2539 & OOM   & OOM & OOM  & OOM \\
        LayerDAG                       & 2539 & 0.34 & 326  & 6e-4 & 1.11 \\
        DiGress                        & 2539 & 16.53 & 326 & 1.51 & 9.42 \\
        DeFoG                          & 2539 & 13.27 & 326 & 1.79 & 9.13 \\
        Directo-DD RRWP                & 2539 & 22.86 & 326 & 1.67 & 19.67 \\
        Directo-DD MagLap              & 2539 & 32.39 & 326 & 3.41 & 20.07 \\
        Directo RRWP                   & 2539 & 22.79 & 326 & 3.23 & 19.55 \\
        Directo MagLap                 & 2539 & 32.52 & 326 & 6.91 & 19.93 \\
        PolarFlow                      & 2539 & 3.33 & 326 & 0.14 & 1.13 \\
        \bottomrule
    \end{tabular}

        \par\vspace{3pt}
        {\scriptsize\raggedright
        \textit{Note:} MLE requires no iterative optimization; its training-time entry therefore reports the total runtime required to fit the empirical distributions. The polar--spectral decompositions of the training graphs are precomputed once on CPU and are excluded from the reported per-epoch training time. 
        \par}
\end{table*}

\section{Additional numerical results}
\label[Appendix]{app.complete_results}
This section provides the complete descriptor-level evaluation results. \Cref{tab:directed_structural,tab:directed_real} report MMD values for the in- and out-degree distributions, directed clustering coefficients, and spectral and wavelet features on the synthetic and real-world datasets, respectively. \Cref{tab:directed_structural_node_label} further reports node-label generation performance using node-type MMD, multiplicity-aware triplet precision and recall, Triplet MMD, and RBF MMD. 

\paragraph{Structural generation}
The in- and out-degree descriptors characterize directional connectivity patterns, directed clustering captures local triadic organization, and the spectral and wavelet features summarize global and multiscale graph structure. As shown in \cref{tab:directed_structural,tab:directed_real}, PolarFlow achieves strong descriptor-level alignment, including the lowest in- and out-degree MMDs on Visual Genome and TPU Tiles and the lowest wavelet MMD on Circuits. Its clustering and spectral results on Circuits are also comparable to those of the learned baselines. Across the complete set of descriptors, PolarFlow obtains the lowest average ratio on all three real-world datasets, indicating strong overall alignment across complementary views of directed graph structure.

\begin{table}[h]
\caption{Directed graph generation performance on structural graphs. Results are presented as mean $\pm$ standard deviation across five sampling runs. }
\label{tab:directed_structural}
\centering
\resizebox{\textwidth}{!}{%
\begin{tabular}{l cccc c ccccc}
\midrule
\textbf{Model} & \textbf{In-degree}$\downarrow$ & \textbf{Out-degree}$\downarrow$ & \textbf{Clustering}$\downarrow$ & \textbf{Spectre}$\downarrow$ & \textbf{Wavelet}$\downarrow$ & \textbf{Ratio}$\downarrow$ & \textbf{Valid}$\uparrow$ & \textbf{Unique}$\uparrow$ & \textbf{Novel}$\uparrow$ & \textbf{V.U.N.}$\uparrow$ \\
\midrule
\multicolumn{11}{c}{Stochastic Block Model (SBM)} \\
\midrule
Training set & 0.0031 & 0.0031 & 0.0274 & 0.0027 & 0.0011 & 1.0 & 97.7 & 100 & 0.0 & 0.0 \\
\midrule
MLE & $0.0020\pm 0.0002$ & $0.0020\pm 0.0002$ & $0.1973\pm 0.0162$ & $0.0087\pm 0.0003$ & $0.0508\pm 0.0009$ & $11.6\pm 0.2$ & $0.0\pm 0.0$ & $100\pm 0.0$ & $100\pm 0.0$ & $0.0\pm 0.0$ \\
DiGress & $0.0038\pm 0.0018$ & $0.0037\pm 0.0018$ & $0.0722\pm 0.0098$ & $0.0045\pm 0.0004$ & $0.0142\pm 0.0039$ & $3.9\pm 0.9$ & $41.5\pm 5.1$ & $100\pm 0.0$ & $100\pm 0.0$ & $41.5\pm 5.1$ \\
DeFoG & $0.0022\pm 0.0016$ & $0.0026\pm 0.0019$ & $0.0813\pm 0.0058$ & $0.0048\pm 0.0002$ & $0.0110\pm 0.0019$ & $4.3\pm 0.8$ & $37\pm 6.6$ & $100\pm 0.0$ & $100\pm 0.0$ & $37.0\pm 6.6$ \\
\midrule
Directo-DD RRWP & $0.0036\pm 0.0018$ & $0.0036\pm 0.0019$ & $0.0592\pm 0.0027$ & $0.0038\pm 0.0007$ & $0.0027\pm 0.0009$ & $ 1.7\pm 0.4$ & $81.5\pm 3.2$ & $100\pm 0.0$ & $100\pm 0.0$ & $81.5\pm 3.2$ \\
Directo-DD MagLap & $0.0038\pm 0.0018$ & $0.0037\pm 0.0019$ & $0.0450\pm 0.0037$ & $0.0038\pm 0.0006$ & $0.0021\pm 0.0009$ & $1.5\pm 0.4$ & $95.5\pm 3.7$ & $100\pm 0.0$ & $100\pm 0.0$ & $95.5\pm 3.7$ \\
Directo RRWP & $0.0028\pm 0.0013$ & $0.0031\pm 0.0014$ & $0.0594\pm 0.0027$ & $0.0035\pm 0.0004$ & $0.0027\pm 0.0009$ & $1.8\pm 0.5$ & $99.5\pm 1.0$ & $100\pm 0.0$ & $100\pm 0.0$ & $99.5\pm 1.0$ \\
Directo MagLap & $0.0038\pm 0.0010$ & $0.0039\pm 0.0012$ & $0.0654\pm 0.0052$ & $0.0038\pm 0.0003$ & $0.0039\pm 0.0008$ & $2.0\pm 0.3$ & $96.5\pm 2.5$ & $100\pm 0.0$ & $100\pm 0.0$ & $96.5\pm 2.5$ \\
\midrule
\rowcolor{gray!12}
Ours & $0.0035\pm 0.0013$ & $0.0037\pm 0.0014$ & $0.0574\pm0.0010 $ & $0.0042\pm 0.0003$ & $0.0048\pm0.0012 $ & $2.0\pm0.3 $ & $92.5\pm 3.8$ & $100\pm 0.0$ & $100\pm 0.0$ & $92.5\pm 3.8$ \\

\midrule
\multicolumn{11}{c}{Twitter} \\
\midrule
Training set & 0.0034     & 0.0017     & 0.0027 & 0.0103 & 0.0086 & 1.0 & --- & 100 & 0.0 & 0.0 \\
\midrule
MLE     & $0.0091\pm0.0001$  & $0.0088\pm0.0001 $ & $0.0919\pm0.0022 $ & $0.0511\pm0.0003$ & $0.1381\pm0.0008$ & $12.59\pm0.18$ & --- & $100\pm0.0$ & $100\pm0.0$ & $100\pm0.0$ \\
DiGress & $0.0023\pm0.0001$  & $0.0035\pm 0.0002$ & $0.0049\pm0.0000 $ & $0.0085\pm0.0002$ & $0.0144\pm0.0005$ & $1.34\pm0.02$ & --- & $100\pm0.0$ & $100\pm0.0$ & $100\pm0.0$ \\
DeFoG   & $0.0037\pm0.0003$  & $0.0024\pm 0.0002$ & $0.0049\pm0.0000 $ & $0.0174\pm0.0014$ & $0.0088\pm0.0005$ & $1.48\pm0.04$ & --- & $100\pm0.0$ & $100\pm0.0$ & $100\pm0.0$ \\
\midrule
Directo-DD RRWP & $0.0041\pm0.0001$  & $0.0016\pm0.0001 $ & $0.0048\pm0.0000 $ & $0.0089\pm0.0004$ & $0.0078\pm0.0003$ & $1.14\pm0.21$ & --- & $100\pm0.0$ & $100\pm0.0$ & $100\pm0.0$ \\
Directo-DD MagLap & $0.0052\pm0.0005$  & $0.0017\pm0.0002 $ & $0.0049\pm0.0000 $ & $0.0153\pm0.0026$ & $0.0063\pm0.0007$ & $1.31\pm0.11$ & --- & $100\pm0.0$ & $100\pm0.0$ & $100\pm0.0$ \\
Directo RRWP & $0.0045\pm0.0003$  & $0.0017\pm0.0002 $ & $0.0048\pm0.0000 $ & $0.0112\pm0.0007$ & $0.0048\pm0.0004$ & $1.15\pm0.04$ & --- & $100\pm0.0$ & $100\pm0.0$ & $100\pm0.0$ \\
Directo MagLap & $0.0052\pm0.0004$  & $0.0019\pm0.0003 $ & $0.0048\pm0.0000 $ & $0.0153\pm0.0021$ & $0.0078\pm0.0008$ & $1.37\pm0.11$ & --- & $100\pm0.0$ & $100\pm0.0$ & $100\pm0.0$ \\
\midrule
\rowcolor{gray!12}
Ours & $0.0030\pm0.0002$  & $0.0026\pm0.0003 $ & $0.0050\pm0.0000$ & $0.0079\pm0.0005$ & $0.0134\pm0.0011$ & $1.33\pm0.04$ & --- & $100\pm0.0$ & $100\pm0.0$ & $100\pm0.0$ \\
\bottomrule
\end{tabular}%
}
\end{table}

\begin{table}[h]
\caption{Directed graph generation performance on real graphs. Results are presented as mean $\pm$ standard deviation across five sampling runs. OOT indicates that the model could not be run within a reasonable time frame.}
\label{tab:directed_real}
\centering
\resizebox{\textwidth}{!}{%
\begin{tabular}{l cccc c ccccc}
\midrule
\textbf{Model} & \textbf{In-degree}$\downarrow$ & \textbf{Out-degree}$\downarrow$ & \textbf{Clustering}$\downarrow$ & \textbf{Spectre}$\downarrow$ & \textbf{Wavelet}$\downarrow$ & \textbf{Ratio}$\downarrow$ & \textbf{Valid}$\uparrow$ & \textbf{Unique}$\uparrow$ & \textbf{Novel}$\uparrow$ & \textbf{V.U.N.}$\uparrow$ \\
\midrule
\multicolumn{11}{c}{Circuit} \\
\midrule
Training set & 0.0005 & 0.0105 & 0.0393 & 0.0123 & 0.0089 & 1.0 & --- & 100 & 0.0 & 0.0 \\
\midrule
MLE & $0.0121\pm0.0029$  & $0.0460\pm0.0070 $ & $0.1701\pm0.0335 $ & $0.1489\pm0.0088$ & $0.0815\pm0.0095$ & $6.49\pm0.65$ & --- & $100.0\pm0.0$ & $100.0\pm0.0$ & $100.0\pm0.0$ \\
DiGress & $0.0266\pm0.0031$  & $0.0117\pm0.0020 $ & $0.0758\pm0.0128 $ & $0.0157\pm0.0038$ & $0.0143\pm0.0023$ & $2.10\pm0.13$ & --- & $100.0\pm0.0$ & $98.46\pm1.88$ & $98.46\pm1.88$ \\
DeFoG & $0.0256\pm0.0024$  & $0.0196\pm0.0023 $ & $0.0734\pm0.0143 $ & $0.0126\pm0.0019$ & $0.0115\pm0.0025$ & $1.90\pm0.24$ & --- & $100.0\pm0.0$ & $99.23\pm1.54$ & $99.23\pm1.54$ \\
\midrule
Directo-DD RRWP & $0.0134\pm0.0026$  & $0.0256\pm0.0037 $ & $0.0678\pm0.0217 $ & $0.0188\pm0.0104$ & $0.0124\pm0.0067$ & $1.95\pm0.44$ & --- & $99.23\pm1.54$ & $96.92\pm2.88$ & $96.17\pm3.21$ \\
Directo-DD MagLap & $0.1241\pm0.0023$  & $0.0272\pm0.0040 $ & $0.0828\pm0.0056 $ & $0.0201\pm0.0095$ & $0.0134\pm0.0062$ & $2.06\pm0.38$ & --- & $100.0\pm0.0$ & $100.0\pm0.0$ & $100\pm0.0$ \\
Directo RRWP & $0.0108\pm0.0030$  & $0.0257\pm0.0030 $ & $0.0852\pm0.0183 $ & $0.0179\pm0.0037$ & $0.0114\pm0.0022$ & $1.90\pm0.21$ & --- & $100\pm0.0$ & $99.23\pm1.53$ & $99.23\pm1.53$ \\
Directo MagLap & $0.0115\pm0.0029$  & $0.0250\pm0.0023 $ & $0.0693\pm0.0099$ & $0.0149\pm0.0030$ & $0.0106\pm0.0030$ & $1.77\pm0.25$ & --- & $99.24\pm1.53$ & $96.91\pm2.87$ & $96.17\pm3.21$ \\
\midrule
Ours & $0.0110\pm0.0008$  & $0.0202\pm0.0034 $ & $0.0657\pm 0.0198$ & $0.0168\pm0.0065$ & $0.0092\pm0.0019$ & $1.61\pm0.22$ & --- & $100\pm0.0$ & $100\pm0.0$ & $100\pm0.0$ \\
\midrule 
\multicolumn{11}{c}{Visual Genome} \\
\midrule
Training set & 0.003 & 0.0018 & 0 & 0.0072 & 0.0036 & 1.0 & 100 & 100 & 0.0 & 0.0 \\
\midrule 
MLE & $0.0474\pm0.0023$  & $0.0607\pm 0.0036$ & $0.5342\pm 0.0830$ & $0.0535\pm0.0015$ & $0.0399\pm0.0017$ & $17.00\pm0.60$ & $90.81\pm0.92$ & $100.0\pm0.0$ & $100.0\pm0.0$ & $90.81\pm0.92$ \\
DiGress & $0.0389\pm0.0032$  & $0.0654\pm 0.0044$ & $0.0008\pm 0.0008$ & $0.0416\pm0.0017$ & $0.0225\pm0.0007$ & $17.00\pm0.60$ & $96.59\pm0.55$ & $100.0\pm0.0$ & $100.0\pm0.0$ & $96.59\pm0.55$ \\
DeFoG & $0.0396\pm0.0028$  & $0.0447\pm 0.0044$ & $0.0002\pm 0.0002$ & $0.0122\pm0.0020$ & $0.0089\pm0.0011$ & $10.60\pm0.80$ & $97.81\pm0.62$ & $100.0\pm0.0$ & $100.0\pm0.0$ & $97.81\pm0.62$ \\
\midrule
Directo-DD RRWP & $0.0413\pm0.0037$  & $0.0424\pm 0.0043$ & $0.0000\pm 0.0000$ & $0.0132\pm0.0053$ & $0.0074\pm0.0083$ & $15.30\pm0.80$ & $99.75\pm0.18$ & $100.0\pm0.0$ & $100.0\pm0.0$ & $99.75\pm0.18$ \\
Directo-DD MagLap & $0.0359\pm0.0036$  & $0.0245\pm 0.0022$ & $0.0000\pm 0.0000$ & $0.0163\pm0.0026$ & $0.0086\pm0.0011$ & $7.60\pm0.70$ & $99.84\pm0.06$ & $100.0\pm0.0$ & $100.0\pm0.0$ & $99.84\pm0.06$ \\
Directo RRWP & $0.0425\pm0.0029$  & $0.0494\pm 0.0021$ & $0.0000\pm 0.0000$ & $0.0226\pm0.0075$ & $0.0228\pm0.0058$ & $12.80\pm0.60$ & $99.94\pm0.05$ & $100.0\pm0.0$& $100.0\pm0.0$ & $99.94\pm0.05$ \\
Directo MagLap & $0.0302\pm0.0040$  & $0.0180\pm 0.0024$ & $0.0000\pm 0.0000$ & $0.0142\pm0.0021$ & $0.0099\pm0.0013$ & $6.02\pm0.51$ & $99.08\pm0.71$ & $99.37\pm0.78$ & $99.68\pm0.63$ & $99.08\pm0.71$ \\
\midrule
\rowcolor{gray!12}
Ours & $0.0089\pm0.0021$  & $0.0055\pm 0.0018$ & $0.1231\pm 0.0446$ & $0.0169\pm0.0028$ & $0.0105\pm0.0009$ & $3.32\pm0.53$ & $91.65\pm1.04$ & $100.0\pm0.0$ & $100.0\pm0.0$ & $91.51\pm0.85$ \\
\midrule 
\multicolumn{11}{c}{TPU Tiles} \\
\midrule
Training set & 0.0008 & 0.0009 & 0.0008 & 0.0009 & 0.0004 & 1.0 & 100 & 100 & 0.0 & 0.0 \\
\midrule
MLE         & $0.1067\pm0.0016$  & $0.0377\pm0.0015 $ & $0.0630\pm0.0031 $ & $0.0681\pm0.0017$ & $0.0752\pm0.0013$ & $95.54\pm0.93$ & $52.27\pm2.32$ & $100.0\pm0.0$ & $100.0\pm0.0$ & $52.27\pm2.32$ \\
D-VAE       &  OOT               &  OOT               &  OOT               &  OOT              &  OOT &  OOT &  OOT &  OOT &  OOT &  OOT \\
LayerDAG    & $0.1754\pm0.0316$  & $0.1377\pm 0.0435$ & $0.0162\pm0.0049 $ & $0.0349\pm0.0064$ & $0.1103\pm0.0206$ & $135.14\pm10.23$ & $100.0\pm0.0$ & $99.93\pm1.22$ & $100.0\pm0.0$ & $99.93\pm1.22$ \\
DiGress     & $0.0953\pm0.0012$  & $0.0056\pm 0.0003$ & $0.0016\pm0.0002 $ & $0.0065\pm0.0010$ & $0.0035\pm0.0006$ & $27.94\pm0.67$ & $87.06\pm1.50$ & $99.88\pm0.25$ & $99.69\pm0.27$ & $86.63\pm1.06$ \\
DeFoG       & $0.0942\pm0.0016$  & $0.0054\pm 0.0003$ & $0.0015\pm 0.0004$ & $0.0045\pm0.0004$ & $0.0027\pm0.0002$ & $26.93\pm0.53$ & $79.69\pm0.76$ & $99.82\pm0.25$ & $99.51\pm0.31$ & $79.20\pm0.85$ \\
\midrule
Directo-DD RRWP & $0.0963\pm0.0013$  & $0.0057\pm 0.0002$ & $0.0029\pm 0.0011$ & $0.0028\pm0.0005$ & $0.0015\pm0.0002$ & $27.02\pm0.47$ & $88.09\pm2.28$ & $93.68\pm1.79$ & $90.18\pm2.22$ & $88.09\pm2.28$ \\
Directo-DD MagLap & $0.0647\pm0.0026$  & $0.0035\pm 0.0003$ & $0.0012\pm 0.0006$ & $0.0029\pm0.0004$ & $0.0018\pm0.0002$ & $18.46\pm0.87$ & $83.74\pm1.63$ & $99.01\pm0.59$ & $97.91\pm0.68$ & $81.74\pm1.69$ \\
Directo RRWP & $0.0952\pm0.0013$  & $0.0056\pm 0.0002$ & $0.0006\pm0.0002 $ & $0.0022\pm0.0001$ & $0.0012\pm0.0002$ & $25.96\pm0.42$ & $86.44\pm1.08$ & $95.83\pm0.53$ & $92.45\pm1.56$ & $78.34\pm1.17$ \\
Directo MagLap & $0.0651\pm0.0019$  & $0.0035\pm 0.0002$ & $0.0077\pm0.0033 $ & $0.0043\pm0.0008$ & $0.0019\pm0.0003$ & $20.96\pm1.35$ & $85.95\pm1.46$ & $98.71\pm0.76$ & $96.81\pm0.49$ & $82.69\pm0.98$ \\
\midrule
\rowcolor{gray!12}
Ours & $0.0004\pm0.0002$  & $0.0028\pm 0.0004$ & $0.0143\pm0.0017 $ & $0.0138\pm0.0012$ & $0.0094\pm0.0008$ & $11.54\pm0.68$ & $84.97\pm1.23$ & $99.81\pm0.15$ & $97.18\pm0.36$ & $82.45\pm1.22$ \\

\bottomrule
\end{tabular}%
}
\end{table}

\paragraph{Node-label generation}  
\Cref{tab:directed_structural_node_label} reports node-label and joint node--edge alignment on the three labeled real-world datasets. Overall, PolarFlow delivers broadly competitive label-generation performance across the complementary metrics. On Circuits, it obtains a competitive node-type MMD and the highest multiplicity-aware precision among the evaluated generators, with recall remaining within the range of the learned baselines. On Visual Genome, it again achieves the highest precision and the lowest RBF MMD while maintaining competitive node-type alignment. On TPU Tiles, its results remain on the same numerical scale as the general-purpose baselines, with Triplet and RBF MMD substantially improving over MLE and LayerDAG. Collectively, these results demonstrate consistent alignment in node-label marginals, semantic-tuple support, and global node--structure representations across different graph domains.

Because node labels are generated conditionally on the sampled adjacency matrix, stronger structural fidelity may provide a more data-like context for label prediction. This is consistent with the favorable RBF MMD on Visual Genome. Since the RBF embeddings jointly encode topology and labels, this result is best interpreted as evidence of end-to-end alignment in the two-stage generation process rather than performance attributable to the label generator alone.

\begin{table}[h]
\caption{Directed graph generation performance on node-label distributional alignment metrics for the three real-world datasets. Results are presented as mean $\pm$ standard deviation across five sampling runs. }
\label{tab:directed_structural_node_label}
\centering
\resizebox{\textwidth}{!}{%
\begin{tabular}{l ccccc}
\midrule
\textbf{Model} & \textbf{Node type (MMD)}$\downarrow$ & \textbf{Precision}$\uparrow$ & \textbf{Recall}$\uparrow$ & \textbf{Triplet MMD}$\downarrow$ & \textbf{RBF}$\downarrow$ \\
\midrule
\multicolumn{6}{c}{Circuit} \\
\midrule
Training set & 0.0057 & 0.8816 & 0.8968  & 0.0117 & 0.0918 \\
\midrule
MLE         & $0.0535\pm0.0015$  & $0.7889\pm0.0284 $ & $0.9062\pm0.0278 $ &  $0.0347\pm0.0023$  & $0.1599\pm0.0334$    \\
Digress     & $0.0951\pm0.0081$  & $0.8793\pm0.0523 $ & $0.8848\pm0.0605 $ &  $0.0177\pm0.0120$  & $0.0973\pm0.0107$    \\
DeFoG       & $0.0885\pm0.0080$  & $0.8682\pm0.0375 $ & $0.8094\pm0.0471 $ &  $0.0213\pm0.0074$  & $0.0983\pm0.0045$    \\
\midrule
Directo-DD RRWP    & $0.0088\pm0.0024$  & $0.8133\pm 0.0454$ & $0.8976\pm0.0148 $ &  $0.0130\pm0.0053$  & $0.0896\pm0.0184$    \\
Directo-DD MagLap  & $0.0039\pm0.0032$  & $0.8751\pm0.0363 $ & $0.8675\pm0.0453 $ &  $0.0147\pm0.0061$  & $0.1032\pm0.0183$    \\
Directo RRWP       & $0.0093\pm0.0099$  & $0.8421\pm0.0440 $ & $0.8529\pm 0.0634$ &  $0.0142\pm0.0058$  & $0.0896\pm0.0097$    \\
Directo MagLap     & $0.0079\pm0.0107$  & $0.8593\pm0.0492 $ & $0.8806\pm 0.0359$ &  $0.0163\pm0.0048$  & $0.0900\pm0.0103$    \\
\midrule
\rowcolor{gray!12}
Ours               & $0.0076\pm0.0037$  & $0.8881\pm 0.0464$ & $0.8697\pm 0.0454$ &  $0.0231\pm0.0089$  & $0.1078\pm0.0137$    \\
\midrule
\multicolumn{6}{c}{Visual Genome} \\
\midrule
Training set & 0.0116 & 0.6381 & 0.6200  & 0.0123 & 0.0256 \\
\midrule
MLE         & $0.0519\pm0.0021$  & $0.6042\pm 0.0307$ & $0.1310\pm 0.0064$ &  $0.3298\pm0.0036$  & $0.1206\pm0.0087$    \\
Digress     & $0.0206\pm0.0034$  & $0.7066\pm 0.0392$ & $0.2748\pm 0.0178$ &  $0.2098\pm0.0097$  & $0.0830\pm0.0120$    \\
DeFoG       & $0.0110\pm0.0019$  & $0.6132\pm 0.0407$ & $0.5309\pm0.0038 $ &  $0.0160\pm0.0036$  & $0.0738\pm0.0024$    \\
\midrule
Directo-DD RRWP    & $0.0231\pm0.0082$  & $0.5583\pm 0.0368$ & $0.5377\pm 0.0291$ &  $0.0272\pm0.0087$  & $0.0533\pm0.0112$    \\
Directo-DD MagLap  & $0.0288\pm0.0069$  & $0.5719\pm 0.0207$ & $0.5188\pm 0.0240$ &  $0.0319\pm0.0054$  & $0.0634\pm0.0111$    \\
Directo RRWP       & $0.0535\pm0.0149$  & $0.4775\pm 0.0462$ & $0.4707\pm 0.0462$ &  $0.0484\pm0.0158$  & $0.0646\pm0.0517$    \\
Directo MagLap     & $0.0299\pm0.0047$  & $0.5008\pm 0.0287$ & $0.4643\pm 0.0246$ &  $0.0348\pm0.0050$  & $0.0802\pm0.0074$    \\
\midrule
\rowcolor{gray!12}
Ours  & $0.0164\pm0.0030$  & $0.7040\pm 0.0192$ & $0.4755\pm0.0213$ &  $0.0701\pm0.0071$  & $0.0458\pm0.0070$    \\
\midrule
\multicolumn{6}{c}{TPU Tiles} \\
\midrule
Training set & 0.00016 & 0.9374 & 0.9769  & 0.00061 & 0.0056 \\
\midrule
MLE         & $0.0010\pm0.0001$  & $0.9051\pm 0.0157$ & $0.9336\pm 0.0054$ &  $0.0914\pm0.0039$  & $0.1578\pm0.0130$    \\
D-VAE       &  OOT   &  OOT   &  OOT   &  OOT   &  OOT      \\
LAYERDAG    & $0.0073\pm0.0027$  & $0.1057\pm 0.0085$ & $0.1228\pm 0.0299$ &  $0.4744\pm0.0110$  & $0.4224\pm0.0767$    \\
Digress     & $0.0014\pm0.0003$  & $0.8946\pm 0.0095$ & $0.9215\pm0.0055 $ &  $0.0047\pm0.0009$  & $0.0160\pm0.0001$    \\
DeFoG       & $0.0007\pm0.0002$  & $0.8976\pm 0.0071$ & $0.9296\pm 0.0110$ &  $0.0041\pm0.0006$  & $0.0167\pm0.0011$    \\
\midrule
Directo-DD RRWP    & $0.0007\pm0.0002$  & $0.9159\pm0.0052 $ & $0.9460\pm0.0051$ &  $0.0022\pm0.0003$  & $0.0126\pm0.0007$    \\
Directo-DD MagLap  & $0.0010\pm0.0003$  & $0.8969\pm0.0111 $ & $0.9349\pm0.0107 $ &  $0.0038\pm0.0009$  & $0.0151\pm0.0004$    \\
Directo RRWP       & $0.0007\pm0.0002$  & $0.9125\pm0.0070$ & $0.9389\pm0.0102 $ &  $0.0024\pm0.0004$  & $0.0132\pm0.0005$    \\
Directo MagLap     & $0.0014\pm0.0002$  & $0.8968\pm0.0082$ & $0.9055\pm 0.0071$ &  $0.0056\pm0.0007$  & $0.0163\pm0.0010$    \\
\midrule
\rowcolor{gray!12}
Ours  & $0.0018\pm0.0004$  & $0.8723\pm 0.0077$ & $0.9067\pm 0.0046$ &  $0.0093\pm0.0006$  & $0.0237\pm0.0021$    \\
\bottomrule
\end{tabular}%
}
\end{table}

\end{document}